\documentclass[11pt,a4paper,reqno]{amsart}

\usepackage{bm}
\usepackage{graphicx}
\usepackage[utf8]{inputenc}
\usepackage[T1]{fontenc}
\usepackage[margin=2.36cm]{geometry}
\usepackage{amsmath}
\usepackage{amsthm}
\usepackage{amssymb}
\usepackage{amsopn}
\usepackage{enumitem}
\usepackage{xcolor}
\usepackage[group-separator={,}]{siunitx}

\usepackage{mdframed}

\usepackage{mathtools}
\usepackage{hyperref}
\usepackage{doi}
\usepackage{xfrac}
\usepackage{graphicx}
\usepackage[export]{adjustbox}
\usepackage{bbm}
\usepackage{relsize}

\usepackage{algpseudocode}
\usepackage{algorithm}

\usepackage[tableposition=top]{caption}
\usepackage[labelformat=simple]{subcaption}

\usepackage[normalem]{ulem}
\usepackage{bm}

\theoremstyle{plain}
\newtheorem{theorem}{Theorem}[section]
\newtheorem{prop}[theorem]{Proposition}

\theoremstyle{definition}

\theoremstyle{remark}
\newtheorem{remark}[theorem]{Remark}

\numberwithin{equation}{section}
\newcommand{\La}{\Lambda}
\newcommand{\R}{\mathbb{R}}
\newcommand{\Rc}{\ensuremath{\mathcal{R}}}

\newcommand{\C}{\mathbb{C}}
\newcommand{\E}{\mathcal{E}}

\newcommand{\<}[2]{\ensuremath{\langle #1,\,#2\rangle}}
\makeatletter
\def\paragraph{\@startsection{paragraph}{4}%
  \z@\z@{-\fontdimen2\font}%
  {\normalfont\bfseries}}
\makeatother

\makeatletter
\def\subsubsection{\@startsection{subsubsection}{3}%
  \z@{.5\linespacing\@plus.7\linespacing}{-.5em}%
  {\normalfont\bfseries}}
  \makeatother

\usepackage{color}

\newcommand{\ar}{a_2}

\graphicspath{{./SAFracture-figs/}}
\begin{document}
\title{A stochastic framework for atomistic fracture}
\author{Maciej Buze}
\address{School of Mathematics, Cardiff University, Senghennyddd Road, CF24 4AG, UK}
\email{BuzeM@cardff.ac.uk}
\author{Thomas E. Woolley}
\address{School of Mathematics, Cardiff University, Senghennyddd Road, CF24 4AG, UK}
\email{WoolleyT1@cardff.ac.uk}
\author{L. Angela Mihai}
\address{School of Mathematics, Cardiff University, Senghennyddd Road, CF24 4AG, UK}
\email{MihaiLA@cardff.ac.uk}
%    General info

\date{\today}


\begin{abstract}
We present a stochastic modeling framework for atomistic propagation of a Mode I surface crack, with atoms interacting according to the Lennard-Jones interatomic potential at zero temperature. Specifically, we invoke the Cauchy-Born rule and the maximum entropy principle to infer probability distributions for the parameters of the interatomic potential. We then study how uncertainties in the parameters propagate to the quantities of interest relevant to crack propagation, namely, the critical stress intensity factor and the lattice trapping range. For our numerical investigation, we rely on an automated version of the so-called numerical-continuation enhanced flexible boundary (NCFlex) algorithm.
\end{abstract}

\maketitle

%{\cb To do: (1) scatter matrix plot upgrade (Figure 8), (2) add third probability curve to Figure 10, (3) check $a_2$ vs $\ar$.}

%%%%%%%%%%%%%%%%%%%%%%%%%%%%%%%%%%%%%%%%%%%%%%%%%%%%%%%%%%%%
%%%%%%%%%%%%%%%%%%%%  NEW SECTION   %%%%%%%%%%%%%%%%%%%%%%%%
%%%%%%%%%%%%%%%%%%%%%%%%%%%%%%%%%%%%%%%%%%%%%%%%%%%%%%%%%%%%
\section{Introduction}
Brittle fracture in crystalline materials is an inherently multiscale phenomenon where macroscopic crack propagation is determined by atomistic processes occurring at the crack tip \cite{Bitzek2015}. The usual modeling approach consists in coupling a high-level, high-accuracy atomistic model, employed in the vicinity of the crack tip, with a continuum-modeled far field. The atomistic model should ideally exhibit a quantum level of accuracy, which can be achieved, for instance, with a Density Functional Theory type model \cite{Kermode2008}. A more computationally feasible alternative is to disregard the electrons and model the inter-atomic interactions instead. In this framework, atoms are treated as points in a discrete model and the behavior of atoms is governed by an empirical (but physics-based) interatomic potential. At this level of description, the two quantities that capture the propagation of a straight crack of single mode are the \emph{critical stress intensity factor} $K_{\rm c}$, and, reflecting the discreteness of the lattice, the \emph{lattice trapping range}, $(K_-,K_+) \subset \R$, which was first identified in \cite{Thomson_1971}.

A typical empirical potential has between $2$ and $11$ parameters (rising to $> 1000$ for modern machine-learning potentials), and the highly nonlinear nature of the overall model necessitates quantifying the uncertainty in their choice and how this propagates to the quantities of interest. In the literature, this is usually done by employing a Bayesian framework, in which one assumes some prior probability distributions for the parameters, which are subsequently updated using available datasets originating from experiments or higher-level theories \cite{Frederiksen2004,Longbottom2019,Wen2020}. However, two main issues can potentially arise with this approach. Firstly, the prior distribution of each parameter is typically taken to be a Gaussian (e.g., in \cite{Frederiksen2004,Longbottom2019,Wen2020}), due to the seemingly reasonable assumption that errors in the reference dataset are independent. This may not necessarily hold true, depending on the physical constraints present in the model, such as, for example, the fact that some parameters have positive values. Secondly, while the Bayesian procedure can be carried out reasonably well for simple quantities of interest, including elastic moduli, lattice parameters, cohesion energy, and point-defect formation energy, it presents a computational challenge for more complicated quantities, such as $(K_-,K_+)$, for which analytical formulae do not exist, and have to be estimated numerically.

Addressing the first issue and inspired by the recent literature on the topic of continuum stochastic elasticity, 
\cite{Soize2006,guill_soize1,Soize2017,GS2020,G2020}, we propose an information-theoretic approach to derive prior distributions for the parameters of the interatomic potential. Our approach uses a minimal set of physical constraints obtained by coupling the atomistic model via the Cauchy-Born rule with its continuum counterpart. To address the computational challenge, we formalize an automated numerical procedure that allows a reasonably fast computation of $K_{\rm c}$ and $(K_-,K_+)$. This is based on a recently proposed NCFlex (numerical continuation-enhanced flexible boundary) scheme \cite{BK2021}.

We demonstrate our approach for an idealized model of straight Mode I fracture in a two-dimensional (2D) crystalline material forming a triangular lattice. In particular, we focus on the Lennard-Jones potential \cite{LJ1924}, which has two parameters, one representing the energetic cost of breaking a bond and the other specifying how difficult it is to break a bond. The relative simplicity of the model ensures that the stochastic framework can be presented with clarity and  allows us to compare our numerical results with some analytical results available in this case. In particular, we show that, in our model, the relative strength of lattice trapping is small and does not depend on the choice of parameters. We further provide evidence that the continuum-theory based formula for $K_{\rm c}$ does hold for our model. We also highlight the interplay between the strength of statistical fluctuations and lattice trapping. In Section~\ref{sec:pre}, we discuss our prerequisites, summarizing the classical continuum framework of linearized elasticity (CLE), and outline the information-theoretic approach, together with a brief account of how this framework translates to fracture modeling. In Section~\ref{sec:deter-atom}, we present the deterministic atomistic model, which uses CLE as a far-field boundary condition. Section~\ref{sec:stoch_atom} is devoted to the development of our atomistic stochastic framework, and is followed by the numerical investigation in Section~\ref{sec:comp}. Some detailed calculations are deferred to Appendix~\ref{appendA}. 

%%%%%%%%%%%%%%%%%%%%%%%%%%%%%%%%%%%%%%%%%%%%%%%%%%%%%%%%%%%%
%%%%%%%%%%%%%%%%%%%%  NEW SECTION   %%%%%%%%%%%%%%%%%%%%%%%%
%%%%%%%%%%%%%%%%%%%%%%%%%%%%%%%%%%%%%%%%%%%%%%%%%%%%%%%%%%%%
\section{Prerequisites}\label{sec:pre}

%%%%%%%%%%%%%%%%%%%%%%%%%%%%%%%%%%%%%%%%%%%%%%%%%
\subsection{Classical linearized elasticity}\label{sec:CLE} 

We consider elastic deformations, ${\bm{y}\,\colon\,\Omega \to\mathbb{R}^3}$, of a three-dimensional body, $\Omega \subset\mathbb{R}^3$, of the form $\bm{y}(\bm{x}) = \bm{x} + \bm{u}(\bm{x})$, where $\bm{u}\,\colon\,\Omega \to \R^3$ is the displacement field. The \emph{strain tensor}, $\bm{\varepsilon}\,\colon\,\Omega \to \R^{3\times 3}$, is defined by ${\bm{\varepsilon}(\bm{x}) = \left[\nabla\bm{u}(\bm{x}) + \nabla\bm{u}(\bm{x})^\top\right]/2}$, and the linear elastic constitutive stress-strain relation \cite{landau1989theory} takes the form
\begin{equation}\label{eq:stress-strain}
\bm{\sigma}(\bm{x}) = \mathbb{C} : \bm{\varepsilon}(\bm{x}),
\end{equation}
where $\mathbb{C} \in\mathbb{R}^{3\times 3 \times 3 \times 3}$ is a constant fourth-order tensor, known as the \emph{elasticity tensor}, and ${\bm{\sigma}\,\colon\,\Omega \to \mathbb{R}^{3\times 3}}$ is the \emph{stress tensor}. In the absence of body forces, equilibrium configurations can be found by solving the equations
\begin{equation}\label{equil-eqn}
\sigma_{ij,j} = 0\,\text{ for }i=1,2,3,
\end{equation}
where the Einstein summation convention is used, subject to appropriate boundary conditions. Depending on the symmetry class considered, the elasticity tensor $\mathbb{C}$ has up to 21 independent entries \cite{landau1989theory}, known as \emph{elasticities}, and admits a second-order tensor representation \cite{Mehrabadi:1990:MC}, in the form of a symmetric matrix $[\mathbb{C}] \in \mathbb{R}^{6\times 6}$. Then the relation \eqref{eq:stress-strain} can be equivalently restated as
\begin{equation}\label{C-matrix}
\begin{bmatrix}
\sigma_{11} \\ \sigma_{22} \\ \sigma_{33} \\ \sigma_{23} \\ \sigma_{31} \\ \sigma_{12}
\end{bmatrix} = \begin{bmatrix}
\C_{1111} & \C_{1122} & \C_{1133} & \C_{1123} & \C_{1131} & \C_{1112}  \\
\C_{2211} & \C_{2222} & \C_{2233} & \C_{2223} & \C_{2231} & \C_{2212}  \\
\C_{3311} & \C_{3322} & \C_{3333} & \C_{3323} & \C_{3331} & \C_{3312}  \\
\C_{2311} & \C_{2322} & \C_{2333} & \C_{2323} & \C_{2331} & \C_{2312}  \\
\C_{3111} & \C_{3122} & \C_{3133} & \C_{3123} & \C_{3131} & \C_{3112}  \\
\C_{1211} & \C_{1222} & \C_{1233} & \C_{1223} & \C_{1231} & \C_{1212}
\end{bmatrix}  
\begin{bmatrix}
\varepsilon_{11} \\ \varepsilon_{22} \\ \varepsilon_{33} \\ 2\varepsilon_{23} \\ 2\varepsilon_{31} \\ 2\varepsilon_{12}
\end{bmatrix}.
\end{equation}
In each symmetry class $[\mathbb{C}]$ can be decomposed as follows,
\begin{equation}\label{C-decomp}
[\mathbb{C}] = \sum_{i=1}^n c_i \bm{E_i},
\end{equation}
where $n$ ranges from $2$ for the fully isotropic case to $21$ when a fully anisotropic case is considered. Here, $\{\bm{E_i}\}$ is the corresponding basis of a subspace of $\R^{6\times 6}$ (see, e.g., \cite{guill_soize1}) and 
\begin{equation}\label{c1cn}
\bm{c} = (c_1,\dots,c_n) \in \R^n
\end{equation}
 is the set of independent entries of $\mathbb{C}$. 
 
 %%%%%%%%%%%%%%%%%%%%%%%%%%%%%%%%%%%%%%%%%%%%%%%%%%%%%5
\subsection{An information-theoretic approach for continuum elastic materials}\label{sec:stoch-cont}

To account for the inherent uncertainty in elastic material parameters, the elasticity tensor can be modeled as a random variable. At the continuum level, the randomness typically stems from: (i) the presence of uncertainties while modeling the experimental setup in either forward simulations or inverse identification; (ii) the lack of scale separation for heterogeneous random materials, hence resulting in the consideration of mesoscopic apparent properties. In \cite{guill_soize1}, a least-informative stochastic modeling framework for elastic material is developed by invoking the maximum entropy principle (MaxEnt) \cite{Jaynes1957}. The minimal set of constraints to explicitly construct a MaxEnt prior probability distribution for $\mathbb{C}$ is as follows:
\begin{enumerate}[label=(\emph{P\arabic*})]
\item The mean value of the tensor is known; 
\item The elasticity tensor $\mathbb{C}$, as well as its inverse, known as the \emph{compliance tensor}, both have a finite second-order moment (physical consistency).
\end{enumerate}
A more detailed treatment of this class of approaches can be found in \cite{Soize2006,Soize2017,GS2020,G2020}.

Given the decomposition of $[\mathbb{C}]$ in \eqref{C-decomp}, the object of interest is a $\mathbb{R}^n$-valued random variable $\bm{c}$, and the constraints \emph{(P1)-(P2)} (together with the required normalization) take the form of a mathematical expectation
\begin{equation}\label{E-f}
\mathbb{E}\{\bm{f}(\bm{c})\} = \bm{h},
\end{equation}
where $\bm{f}\,\colon\,\R^n \to \R^q$ and $\bm{h} \in \R^q$. It can be shown \cite{Jaynes1957,Mehta2004} that the MaxEnt probability distribution of the random variable $\bm{c}$ from \eqref{c1cn} is characterized by the probability density function
\begin{equation}\label{rho-general}
\rho(\mathfrak{c}) := \bm{1}_{S}(\mathfrak{c}) \exp\{- \<{\bm{\lambda}}{\bm{f}(\mathfrak{c})}_{\mathbb{R}^q}\},
\end{equation}
where the set $S \subset \mathbb{R}^n$ represents all possible choices of $\bm{c} \in \mathbb{R}^n$ for which \eqref{E-f} is satisfied, whereas $\bm{\lambda} = (\lambda_1,\dots,\lambda_q)$ is the vector of the associated Lagrange multipliers. We refer to \cite{guill_soize1} for an in-depth discussion, and in particular, to their Appendix~B, in which the existence and uniqueness of a MaxEnt probability density function is addressed,  explaining why this takes the form \eqref{rho-general}.

%%%%%%%%%%%%%%%%%%%%%%%%%%%%%%%%%%%%%%%%%%%%%%%%%%%%%%%%%%
\subsection{Mode I fracture in planar elasticity}\label{sec:cle-in-plane}

Our stochastic framework for atomistic crack propagation will be presented for the case of a single Mode I crack in a cubic crystal modeled in the in-plane approximation (see Figure~\ref{fig:domain_cont}). This leads to considerable simplification of the general theory presented in Section~\ref{sec:CLE}. 

%%%%%%%%%%%%%
\begin{figure}[htbp]
	\includegraphics[width=0.95\textwidth]{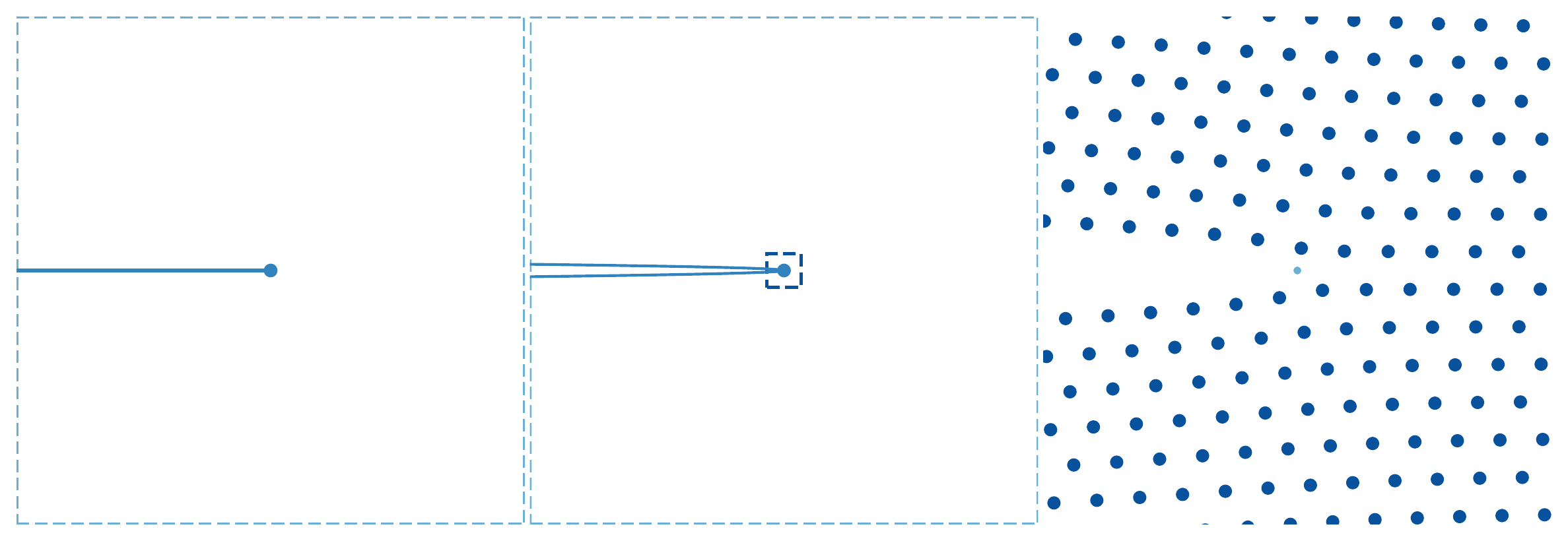}
	\caption{Left: A continuum $\mathbb{R}^2 \setminus \Gamma_0$ in the reference configuration. Middle: The configuration of a cracked body obtained from the equilibrium displacement field $\bm{u}$ from \eqref{u_CLE}, with the vicinity of the crack tip (in which atomistic effects dominate) highlighted. Right: The atomistic structure in the vicinity of the crack tip.}
	\label{fig:domain_cont}
\end{figure}
%%%%%%%%%%%%

In \emph{planar elasticity} \cite{Thorpe1992} (see \cite[Appendix]{Ostoja2002} for a discussion on the plane-strain and plane-stress reductions of the three-dimensional elasticity theory),  the strain components are $\varepsilon_{11}, \varepsilon_{22}, \varepsilon_{12}$ and the stress components are $\sigma_{11},\sigma_{22}, \sigma_{12}$. In combination with the fact that in the cubic symmetry class $[\mathbb{C}]$ has three independent entries $\bm{c} = (c_1,c_2,c_3)$, equation \eqref{C-matrix} simplifies to
 \[
\bm{\sigma}_{\rm 2D} = \mathbb{C}_{\rm 2D}:\bm{\varepsilon} \iff \begin{bmatrix}
\sigma_{11} \\ \sigma_{22} \\ \sigma_{12}
\end{bmatrix} = \begin{bmatrix}
c_1 & c_2 & 0  \\
c_2 & c_1 & 0  \\
0 & 0 & c_3
\end{bmatrix}  
\begin{bmatrix}
\varepsilon_{11} \\ \varepsilon_{22} \\ 2\varepsilon_{12}
\end{bmatrix}.
\]
As will become apparent when we introduce the atomistic setup in Section \ref{sec:deter-atom}, we consider in fact a special case, such that
\begin{equation}\label{c1c2c3}
c_2 = c_3\quad\text{ and }\quad c_3 = \frac{c_1-c_2}{2} \implies c_2 = \frac{1}{3}c_1 = \mu,
\end{equation}
where $\mu$ denotes the \emph{shear modulus} and represents the only independent entry of the elasticity tensor.  In this case, an equilibrium displacement field, $\bm{u}\,\colon\,\R^2\setminus\Gamma_0 \to \R^2$, around a Mode I crack, with the crack surface described by
\[
\Gamma_0 = \{\bm{x}=(x_1,x_2) \in \R^2\,\mid\, x_1 < 0\;\text{ and }\; x_2 = 0\},
\]
and which satisfies the equilibrium equations \eqref{equil-eqn} subject to homogeneous Neumann boundary condition on $\Gamma_0$, can be shown \cite{SJ12} to be given by
\begin{equation}\label{u_CLE}
K\widehat{\bm{u}}(\bm{m})=  \frac{K}{4\sqrt{2\pi}\mu}\sqrt{r} \begin{bmatrix} 3\cos(\theta/2) - \cos(3\theta/2) \\ 5\sin(\theta/2) - \sin(3\theta/2) \end{bmatrix},\qquad \theta\in(-\pi,\pi),
\end{equation}
where we employ polar coordinates $\bm{m} = r(\cos\theta,\sin\theta)$ and $K \in \R$ is the \emph{stress intensity factor} and enters as a prefactor.

According to \emph{Griffith's criterion} \cite{SJ12}, at the continuum level of description, there exists a critical $K_{\rm c}$, so that, when $K > K_{\rm c}$, it is energetically favorable for the crack to propagate. It can be shown \cite{Zehnder2012} that, in the case considered, the critical value is
\begin{equation}\label{Kc}
\widetilde{K}_{\rm cont} = 4\sqrt{\frac{\gamma\mu}{3}},
\end{equation}
where $\gamma$ is the surface energy per unit area, which is a material-dependent quantity. 

It is well-known \cite{Bitzek2015}, however, that the continuum picture is incomplete in the case of brittle fracture in crystalline materials, and one should not omit atomistic effects occurring at the crack tip. We will proceed to present the atomistic framework.

%%%%%%%%%%%%%%%%%%%%%%%%%%%%%%%%%%%%%%%%%%%%%%%%%%%%%%%%%%%%
%%%%%%%%%%%%%%%%%%%%  NEW SECTION   %%%%%%%%%%%%%%%%%%%%%%%%
%%%%%%%%%%%%%%%%%%%%%%%%%%%%%%%%%%%%%%%%%%%%%%%%%%%%%%%%%%%%
\section{Deterministic atomistic setup}\label{sec:deter-atom}

In this section, we introduce the atomistic setup by recalling well-established arguments setting out why the continuum picture is insufficient, followed by a detailed discussion on discrete kinematics, Cauchy-Born rule and atomistic fracture. 

%%%%%%%%%%%%%%%%%%%%%%%%%%%%%%%%
\subsection{Lattice trapping}\label{lat-trap}

Cracks in brittle materials are known to propagate via atomistic mechanisms involving breaking of chemical bonds between atoms at the crack tip \cite{Bitzek2015}. In particular, as first reported in \cite{Thomson_1971} and confirmed for a model similar to ours in \cite{Sinclair}, the discreteness of the lattice implies that the crack remains locally stable for a range of stress intensity factors
\begin{equation}\label{Krange}
I = (K_-,K_+) \subset \mathbb{R},
\end{equation}
also known as the \emph{lattice trapping range}. This is in contrast with the continuum theory outlined in Section~\ref{sec:cle-in-plane}. At the atomistic level, the critical
\begin{equation}\label{Kat}
\widetilde{K}_{\rm at} \in I
\end{equation}
corresponds to a unique value for which the atomistic energy is the same both prior and after the crack propagating by one lattice spacing $\ell$. 

%%%%%%%%%%%%%
\begin{figure}[htbp]
\includegraphics[width=0.9\textwidth]{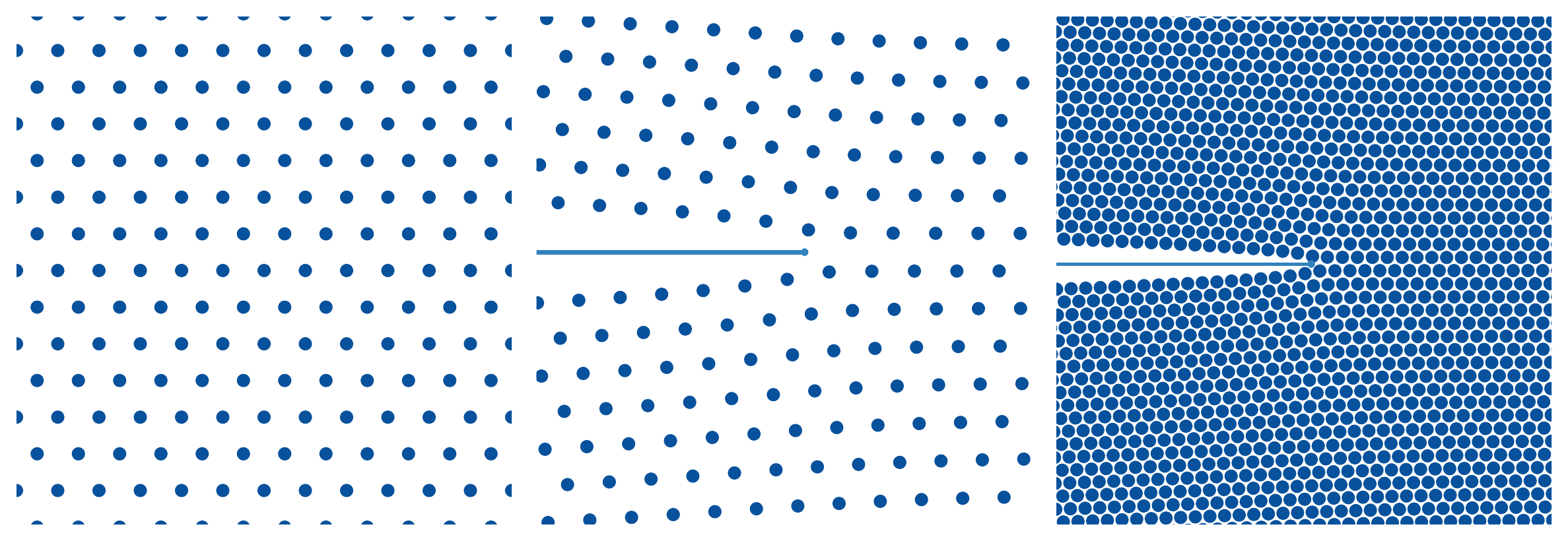}
\caption{Left: A defect-free crystalline material $\bm{\La}$. Middle: The vicinity of a Mode I crack tip with the crack surface depicted as a dotted line. Right: Zoomed-out view of the cracked crystal.}
\label{fig:domain}
\end{figure}
%%%%%%%%%%%%

%%%%%%%%%%%%%%%%%%%%%%%%%%%%%%%%%
\subsection{Discrete kinematics}\label{setup}

We consider a 2D crystalline material, $\bm{\La}$, given by the infinite triangular lattice (see Figure~\ref{fig:domain}) defined by
\begin{equation}\label{lattice}
\bm{\La}= \ell\left({\bm{M}} \mathbb{Z}^2 - \widehat{\bm{x}}\right),\qquad {\bm{M}} = \begin{bmatrix} 1 & \frac{1}{2} \\ 0 & \frac{\sqrt{3}}{2} \end{bmatrix},\qquad \widehat{\bm{x}} = \begin{bmatrix} \frac{1}{2} \\\frac{\sqrt{3}}{4} \end{bmatrix},
\end{equation}
where $\mathbb{Z}^2 = \{(m_1,m_2) \in \R^2 \mid m_1,m_2 \in \mathbb{Z}\}$ and the prefactor $\ell > 0$ is the so-called lattice constant, describing the natural distance between atoms in the material, and can be measured experimentally. Conceptually, the 2D domain is to be interpreted as a cross-section of a three-dimensional material body, which is periodic in the anti-plane direction. For instance, the triangular lattice is known to be obtained as a projection of the body-centered-cubic lattice \cite[Figure 1]{2017-bcscrew}, which is a crystalline arrangement that can be found in many real-world materials \cite{Kittel1996}. 

The atoms are assumed to interact within a finite interaction range $\Rc \subset \bm{\La} \setminus \{\bm{0}\}$, which is assumed to respect lattice symmetries, enforced through defining 
\begin{equation}\label{Rc}
\Rc = (\bm{\La} \setminus \{\bm{0}\}) \cap B_{\ell R_*},
\end{equation} 
for some $R_* > 0$, where $B_R$ is the ball of radius $R$ centred at the origin (see Figure \ref{fig:interactions}). The rescaling by the lattice constant $\ell$ ensures that $R_*$ is an independent parameter, in the sense that $R_*$ uniquely determines the number of atoms in the interaction radius, regardless of the lattice constant.

As noted in Section \ref{sec:cle-in-plane}, we are interested in the in-plane deformations of the material, described by a function $\bm{y} \,:\,\bm{\La} \to \mathbb{R}^2$, and we will use the notation $\bm{y}^{\bm{U}}$ for
\[
\bm{y}^{\bm{U}}(\bm{m}) = \bm{m} + \bm{U}(\bm{m}),
\] 
where ${\bm{U} \,:\,\La \to \mathbb{R}^2}$ is the displacement. For any $\bm{m} \in\bm{\La}$ and $\bm{\rho} \in {\Rc}$, the finite difference of the deformation at sites $\bm{m}$ and $\bm{m}+\bm{\rho}$ is defined as $D_{\bm{\rho}}\bm{y}(\bm{m})= \bm{y}(\bm{m}+\bm{\rho})-\bm{y}(\bm{m})$. The discrete gradient is then
\[
\bm{D}\bm{y}(\bm{m}):= \left(D_{\bm{\rho}}\bm{y}(\bm{m})\right)_{\bm{\rho}\in\Rc} \in \left(\R^{2}\right)^{{\Rc}},
\]
where the convenient ordering short-hand notation $\left(\R^{2}\right)^{{\Rc}}$ refers to the space $\mathbb{R}^{2\times |{\Rc}|}$ where ${|{\Rc}| \in \mathbb{N}}$ is the number of elements in ${\Rc}$. For the identity deformation $\bm{y}^{\bm{0}}(\bm{m}) = \bm{m}$, note that \linebreak ${\bm{D}\bm{y}^{\bm{0}}(\bm{m}) = (\bm{\rho})_{\bm{\rho} \in \Rc}}$ and so we will sometimes use the notation $(\bm{\rho})$.

The interaction between atoms is encoded in an interatomic potential $V\,:\, \left(\R^2\right)^{{\Rc}} \to \mathbb{R}$ with a site energy given by ${V}(\bm{D}\bm{y}(\bm{m}))$. In the present study, $V$ is restricted to be a pair potential admitting a decomposition of the form 
\begin{equation}\label{V}
{V}(\bm{D}\bm{y}(\bm{m})) = \sum_{\bm{\rho} \in {\Rc}} \phi(|D_{\bm{\rho}}\bm{y}(\bm{m})|),
\end{equation}
and $\phi\,:\,\R \to \R$ is the Lennard-Jones  potential \cite{LJ1924} given by 
\begin{equation}\label{LJ}
 \phi(r) = 4a_1\left[\left(\frac{1}{\ar r}\right)^{12} - \left(\frac{1}{\ar r}\right)^{6}\right],
\end{equation}
with two parameters $a_1,\ar$. Its typical shape is depicted in Figure~\ref{fig:interactions}. We note that, usually, the second parameter $\ar$ is placed in the numerator of the two power terms. As will be discussed in Remark~\ref{rem1}, in our work, it is more convenient to have it introduced as in \eqref{LJ} instead.

The lattice constant $\ell$, from \eqref{lattice}, can be shown to be uniquely determined by $\ar$ and $R_*$ in our model, that is
\begin{equation}\label{lat-const-dep}
\ell(\ar,R_*) = \left(\frac{B_{R_*}}{A_{R_*}}\right)^{1/6} \ar^{-1},
\end{equation}
where the constants $A_{R_*},\,B_{R_*}$ depend on how many neighbors there are in the interaction range. For instance, if we only look at nearest neighbors, i.e., $R_* = 1$, then $A_1 = 1$ and $B_1 = 2$. If $R_* = \sqrt{3}$ (next-to-nearest neighbors also included), then $A_{\sqrt{3}} = 656/27$ and $B_{\sqrt{3}} = 35008/729$. The relevant calculations are outlined in Appendix~\ref{appendA}. 

The resulting energy of the system is formally given by 
\begin{equation}\label{energy-formal}
\mathcal{E}(\bm{U}) = \sum_{\bm{m} \in \bm{\La}} V(\bm{D}\bm{y}^{\bm{U}}(\bm{m})).
\end{equation}

%%%%%%%%%%%%%%
\begin{figure}[htbp]
\includegraphics[width=0.9\textwidth]{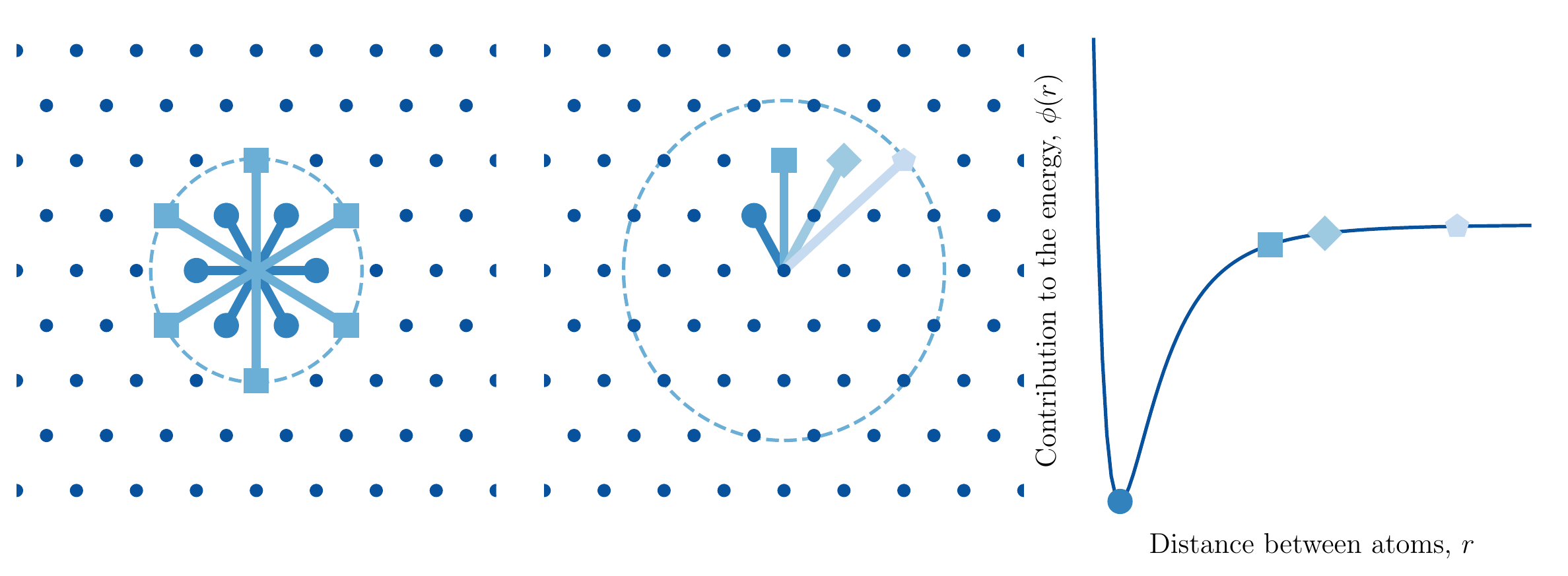}
\caption{Left: The interaction range $\Rc$ from \eqref{Rc} for $R_* = \sqrt{3}$ depicted as a dotted line. All first and second neighbor interaction bonds highlighted. Middle: The interaction range $\Rc$ from \eqref{Rc} for $R_* = \sqrt{7}$ depicted as a dotted line with examples of first, second, third and fourth neighbours shown. Right: Typical contributions of interaction bonds to the energy for the pair potential $\phi$ defined in \eqref{LJ}. Note that when $R_* \geq \sqrt{3}$, it is not typically true for the first neighbour distance to coincide with the minimum of $\phi$.}
\label{fig:interactions}
\end{figure}
%%%%%%%%%%%%

%%%%%%%%%%%%%%%%%%%%%%%%%%%%%%%%%%
\subsection{Cauchy-Born rule}\label{sec:CB}
As investigated in \cite{E84,friesecketheil,emingstatic,ortnertheil13}, for example (see also the most recent survey article \cite{Ericksen08}), a consistent way to link the atomistic model with its continuum counterpart is through the Cauchy-Born rule. In this framework, the interatomic potential, $V$, the interaction range, $\Rc$, and the lattice, $\bm{\La}$, together give rise to a continuum Cauchy-Born strain energy function ${W\,:\, \R^{2 \times 2} \to \R\cup \{+\infty\}}$ through the coupling 
\begin{equation}\label{W-CB}
W(\bm{F}):= \frac{1}{{\rm\det} (\ell\bm{M})} V\left( (\bm{F}\bm{\rho})_{\bm{\rho} \in {\Rc}} \right),
\end{equation}
where $\bm{F} \in \R^{2 \times 2}$ is the displacement gradient arising from the homogeneous displacement field $\bm{U}(\bm{x}) = \bm{F}\bm{x}$. 

A subsequent expansion of $W$ to second order around the identity yields the elasticity tensor $\C$ with 
\begin{equation}\label{C-general}
\C_{i\alpha j\beta}:= \partial_{F_{i\alpha}F_{j\beta}}W(\mathbbm{1}) = \frac{1}{{\rm\det}(\ell\bm{M})} \sum_{\bm{\rho},\bm{\sigma} \in \Rc} \partial^2_{i\bm{\rho} j \bm{\sigma}} V((\bm{\rho})) \rho_{\alpha} \sigma_{\beta}.
\end{equation}
In the case of a pair potential, it further simplifies to 
\begin{equation}\label{C-pp}
\C_{i\alpha j\beta} = \frac{1}{{\rm\det} (\ell\bm{M})} \sum_{\bm{\rho} \in {\Rc}}\left[\left(\frac{\phi''(|\bm{\rho}|)}{|\bm{\rho}|^2} - \frac{\phi'(|\bm{\rho}|)}{|\bm{\rho}|^3}\right)\rho_i\rho_j + \delta_{ij}\frac{\phi'(|\bm{\rho}|)}{|\bm{\rho}|}\right]\rho_{\alpha}\rho_{\beta}.
\end{equation}
Thus, unlike in the continuum linear elasticity setup, where the elasticities are the independent parameters specifying the material model, here, they are derived quantities, and are in effect nonlinear functions of the potential parameters $a_1,\ar$, introduced in \eqref{LJ}, and in principle also of $\ell$ and $R_*$. A calculation presented in Appendix~\ref{appendA} further shows that \eqref{c1c2c3} is indeed satisfied and the shear modulus $\mu$ is given by 
\begin{equation}\label{mu-a1a2}
\mu = D_{R_*} a_1 \ar^{2},
\end{equation}
where $D_{R_*}$ is a known constant depending on $R_*$. 

%%%%%%%%%%%%%%%%%%%%%%%%%%%%%%%%%%%%%%%%%%%%%
\subsection{Mode I atomistic fracture}\label{sec:atom-mode1}
Due to the inherent nonlinearity of the atomistic model, it is not possible to obtain an analytic characterization of atomistic equilibrium configurations around a crack. Away from the crack tip, however, the CLE model outlined in Section~\ref{sec:CLE}, which can be obtained via the Cauchy-Born coupling, as discussed in Section~\ref{sec:CB}, approximates the atomistic model well \cite{2018-antiplanecrack}.

The CLE solution $K\widehat{\bm{u}}$ from \eqref{u_CLE} is thus a suitable far-field boundary condition. We impose this by looking at displacements, 
$\bm{U}\,\colon\,\bm{\La} \to \R^2$, of the form
\[
\bm{U}(\bm{m}) =K\widehat{\bm{u}}(\bm{m} - \bm{\alpha}) + \bm{u}(\bm{m}),
\]
where the near-crack-tip \emph{atomistic correction} $\bm{u}$ is constrained to satisfy 
\begin{equation}\label{Du-Duhat}
|\bm{D}\bm{u}(\bm{m})| \ll |\bm{D}\widehat{\bm{u}}(\bm{m})|\,\text{ when }\,|\bm{m}| \gg 1.
\end{equation} 
This is consistent with the idea presented in the middle panel of Figure \ref{fig:domain_cont}. The horizontal shift is ${\bm{\alpha} = (\alpha,0) \in \R^2}$, where $\alpha \in \R$ is introduced as a variable to be able to track the crack tip position. 

The formally defined infinite lattice energy we wish to equilibrate is given by
\begin{equation}\label{en-infinite}
\E(\bm{u},\alpha,K) = \sum_{\bm{m} \in \bm{\La}} V(\bm{D}\bm{y}^{\bm{U}}(\bm{m})) - V(\bm{D}\bm{y}^{\bm{0}}(\bm{m})),
\end{equation}
where $\bm{y}^{\bm{0}}(\bm{m}) = \bm{m} + K\widehat{\bm{u}}(\bm{m})$. Since in this framework the triplet $(\bm{u},\alpha,K)$ fully determines the displacement $\bm{U}$, we shall often identify $\bm{U} = (\bm{u},\alpha,K)$.

The lattice trapping range $I$ from \eqref{Krange} can be found by tracing continuous paths of solutions $(0,1) \ni s \mapsto \bm{U}_s = (\bm{u}_s,\alpha_s,K_s)$, such that 
\[
(\delta_{\bm{u}} \mathcal{E}(\bm{U}_s),\delta_{\alpha}\mathcal{E}(\bm{U}_s)) = \bm{0}.
\]
As reported in \cite{BK2021} and earlier in \cite{2019-antiplanecrack}, the resulting path $s \mapsto \bm{y}_s$ of equilibrium configurations is expected to be a vertical snaking curve, capturing bond-breaking events, with $K_s$ oscillating within a fixed interval, which is the lattice trapping interval defined in \eqref{Krange} (see Figure~\ref{fig:lat_trap_I} for an example of a numerically computed snaking curve). 

We refer to \cite{2018-antiplanecrack,2019-antiplanecrack} for a rigorous derivation of the infinite lattice model. As will be noted in Section~\ref{sec:comp}, in the present work, we restrict our attention to the case where \eqref{Du-Duhat} is satisfied through setting  $\bm{u}(\bm{m})= \bm{0}$ for all $\bm{m} \in \bm{\La}$, such that $|\bm{m}| > R_0$, for some suitably chosen $R_0$.

We further note that the continuum-theory based prediction for the critical stress intensity factor given by \eqref{Kc} can be computed for the atomistic model, since the shear modulus $\mu$ and the surface energy $\gamma$ can be computed directly from the atomistic model. It is in fact widely assumed that, in the infinite lattice, the critical stress intensity factor in the atomistic description \eqref{Kat} and in the continuum description \eqref{Kc} coincide, that is, 
\begin{equation}\label{Kc-Kat}
\widetilde{K}_{\rm at} = \widetilde{K}_{\rm cont}.
\end{equation}
The numerical work presented in Section~\ref{sec:comp} will, among other things, provide evidence that, in our model, this equality holds, subject to accounting for finite-domain effects.

%%%%%%%%%%%%%%%%%%%%%%%%%%%%%%%%%%%%%%%%%%%%%%%%%%%%%%%%%%%%
%%%%%%%%%%%%%%%%%%%%  NEW SECTION   %%%%%%%%%%%%%%%%%%%%%%%%
%%%%%%%%%%%%%%%%%%%%%%%%%%%%%%%%%%%%%%%%%%%%%%%%%%%%%%%%%%%%
\section{Stochastic atomistic framework}\label{sec:stoch_atom}

%%%%%%%%%%%%%%%%%%%%%%%%%%%%%%%%%%%%%%%%%%%%%%%%%%
\subsection{An information-theoretic formulation for  Lennard-Jones potential}\label{sec:stoch-LJ}

We aim to quantify how the uncertainty in the choice in the model parameters propagates to the computed quantities of interest (QoI), which in the case of atomistic fracture are
\begin{equation}\label{QoI}
{\rm {(QoI)}}\quad I, \; \widetilde{K}_{\rm at} \text{ and }\widetilde{K}_{\rm cont}.
\end{equation}

Inspired by the corresponding work in the continuum setup from \cite{guill_soize1}, outlined in Section \ref{sec:stoch-cont}, we invoke the MaxEnt to infer the probability distributions of parameters present in the model, that is
\[
(a_1,\ar, \ell, R_*) \in \R^{4},
\]
where we recall that $a_1$ and $\ar$ are the potential parameters introduced in \eqref{LJ}, $\ell$ is the lattice constant introduced in \eqref{lattice}, $R_*$ is the interaction radius from \eqref{Rc}. 

As noted in \eqref{lat-const-dep}, for the Lennard-Jones potential $\phi$ defined by \eqref{LJ}, the lattice constant $\ell$ is uniquely determined by $\ar$ and $R_*$, so $\ell$ is not an independent parameter.

We further recognize the special nature of the interaction radius parameter $R_*$, which is not a parameter that would typically be considered as a random variable, but rather fixed a priori. Even if it was to be modeled as a random variable, and we note that the information-theoretic stochastic framework provides us with a way of doing so, it would be effectively a countable random variable. For the purpose of analysis, in this section, we consider $R_*$ fixed and later in the numerical section we will consider three deterministic choices for $R_*$, corresponding to including interaction with up to first, second and third nearest neighbors, respectively (see Figure~\ref{fig:interactions}).

We gather the remaining independent parameters as
\begin{equation}\label{A-ran}
\bm{A} = (a_1, \ar) \in \R^2.
\end{equation}
Recalling the set of natural constraints \emph{(P1)-(P2)} in Section \ref{sec:stoch-cont}, we first restate \emph{(P1)} as 
\begin{equation}\label{mean}
\text{\emph{(P1)}: }\;\mathbb{E}(\bm{A}) = \underline{\bm{a}},
\end{equation}
where $\underline{\bm{a}} = (\underline{a_1},\underline{\ar})$ is known, corresponding to default parameters of the potential. The second constraint \emph{(P2)} concerns the elasticity tensor, which, through the Cauchy-Born rule discussed in Section \ref{sec:CB} and the underlying assumption of planar elasticity and the pairwise nature of the interatomic potential, simplifies so that the only independent parameter is the shear modulus $\mu$, which, as established in \eqref{mu-a1a2}, is a function of $\bm{A}$, is the only independent elastic constant. This leads us to recast \emph{(P2)} as
\begin{equation}\label{logmu}
\text{\emph{(P2):} }\;\mathbb{E}(\log(\mu(\bm{A}))) = \nu_{\bm{A}},
\end{equation}
where $\nu_{\bm{A}}$ is a given parameter such that $|\nu_{\bm{A}}| < \infty$. For the rationale as to why the condition of this type ensures \emph{(P2)} we refer to \cite{guill_soize1}.

\begin{prop}\label{prop1}
Under the constraints \eqref{mean} and \eqref{logmu}, the MaxEnt probability density function of the random variable $\bm{A}$ defined in \eqref{A-ran} is given by
\[
\rho_{\bm{A}}(\bm{a}) =  \rho_{A_1}(a_1) \times  \rho_{A_2}(\ar),
\]
where
\[
\rho_{A_1}(a_1) = \mathbbm{1}_{\R_+}(a_1)k_1 a_1^{-\tau}\exp\{-\lambda_{1}a_1\}
\]
and
\[
\rho_{A_2}(\ar) = \mathbbm{1}_{\R_+}(\ar)k_2\ar^{-2\tau}\exp\{-\lambda_{2}\ar\},
\]
with $k_1$ and $k_2$ positive normalization constants, and $\lambda_{1}$ and $\lambda_{2}$ Lagrange multipliers corresponding to \emph{(P1)}. The parameter $\tau$ controls the level of statistical fluctuations and is required to satisfy $\tau \in (-\infty,1/2)$.

It follows that $a_1$ and $\ar$ are statistically independent, with $a_1$ Gamma-distributed with shape and scale hyperparameters $(\alpha_{1},\beta_{1}) = \left(1-\tau,\underline{a_1}/(1-\tau)\right)$ and $\ar$ Gamma-distributed with shape and scale hyperparameters $(\alpha_{2},\beta_{2}) = \left(1{-2\tau},\underline{\ar}/(1{-2\tau})\right)$. 
\end{prop}
\begin{proof}
The constraints in \eqref{mean} and \eqref{logmu}, together with the normalization constraint, can be put in the form of a mathematical expectation as in \eqref{E-f}, namely
\[
\mathbb{E}\{\bm{g}(\bm{A})\} = \bm{\hat{g}},
\]
where $\bm{g}\,\colon\,\R^2 \to \R^4$ with $\bm{g}(\bm{A}) = (\bm{A}, \log(\mu(\bm{A})),1) \in \R^4$ and  $\bm{\hat{g}} = (\underline{\bm{a}},\nu_{\bm{A}},1) \in \R^4$. It follows from \eqref{rho-general} that
\[
\rho_{\bm A}(\bm{a}) = \bm{1}_{R^2_+}(\bm{a}) \exp\{- \<{\bm{\lambda}}{\bm{g}(\bm{a})}_{\R^4}\}, 
\]
as $R^2_+$ is the largest set on which \eqref{logmu} is satisfied. Since
\[
\exp\{-\<{\bm{\lambda}}{\bm{g}(\bm{a})}_{\R^4}\} = k_0\exp(-\lambda_1 a_1)a_1^{-\lambda_3}\exp(-\lambda_2 \ar)\ar^{-2\lambda_3},
\]
where $k_0 = \exp(-\lambda_4)D_{R_*}^{-\lambda_3}$, the result follows by identifying $\lambda_3 = \tau$, and an appropriate splitting of the normalization constant as $k_0 = k_1 k_2$. 
\end{proof}
\begin{remark}\label{rem1}
The Lennard-Jones potential $\phi$ defined \eqref{LJ} is typically introduced with the second parameter $a_2' := \ar^{-1}$. From the information-theoretic point of view it is far less convenient to do so, as then the MaxEnt distribution of $a'_2$ can be shown (using the framework discussed in this section) to be the Gamma distribution with shape and scale hyperparameters $(1+2\tau, \underline{a'_2}/(1+2\tau))$, provided that $\tau \in (-\frac{1}{2},+\infty)$. Thus, the setup where both $a_1$ and $a_2'$ follow the Gamma distribution would only apply when $\tau \in (-\frac{1}{2},1)$, which is more restrictive than what we obtain in Proposition~\ref{prop1}.
\end{remark}

For the model under consideration, the following can be subsequently established about $\widetilde{K}_{\rm cont}$. 
\begin{prop}\label{prop2}
The critical stress intensity factor $\widetilde{K}_{\rm cont}$, when computed for the shear modulus $\mu$ and the surface energy $\gamma$ obtained directly from the atomistic model, satisfies
\[
\widetilde{K}_{\rm cont} = C_{R_*} a_1\,\ar^{3/2},
\]
where the constant $C_{R_*}$ depends only on the interaction range $R_*$.

If $\bm{A}$ is taken to follow the MaxEnt distribution established in Proposition~\ref{prop1}, then $\widetilde{K}_{\rm cont}$ is a random variable with the probability density function $\rho_{\widetilde{K}_{\rm cont}}$ given by
\[
\rho_{\widetilde{K}_{\rm cont}}(k) = \frac{1}{C_{R_*}}\int_{\R}\ar^{-3/2} \rho_{\bm A}\left(\frac{k}{C_{R_*}\ar^{3/2}},\ar\right)d\ar.
\]
\end{prop}
\begin{proof}
At the atomistic level of description, the energetic cost of creating a surface is equivalent to the energetic cost of breaking interaction bonds between atoms on opposite sides of the surface.

We assume first that $R_* = 1$, that is, we only look at the nearest neighbor interaction. In this case, the lattice constant $\ell$ minimizes the potential $\phi$, and in fact $\phi(\ell) = -a_1$. The cost of breaking one bond is then
\[
\lim_{r\to \infty}\phi(r) - \phi(\ell) = -\phi(\ell) = a_1.
\]
When the crack surface is extended by length $L$, on the triangular lattice, this corresponds to breaking interaction $2L/\ell$ bonds. Then the surface energy per unit area $\gamma$ from \eqref{Kc} is given by 
\[
\gamma = \frac{1}{L}\frac{2L}{\ell}\left(\lim_{r\to \infty}\phi(r) - \phi(\ell) \right) = \tilde{C}_1  a_1 \ar,
\]
where $\tilde{C}_1 = 2^{5/6}$. This follows from \eqref{lat-const-dep}.

It is shown in Appendix~\ref{appendA} that, in the case of a general $R_*$, we have
\begin{equation}\label{gamma-gen}
\gamma = \tilde{C}_{R_*} a_1 \ar.
\end{equation}
Using \eqref{Kc} and \eqref{mu-a1a2}, we then arrive at
\[
\widetilde{K}_{\rm cont} = \frac{4}{\sqrt{3}} \sqrt{\gamma\mu} = C_{R_*}  \sqrt{a_1} \sqrt{\ar} \sqrt{a_1} \ar = C_{R_*} a_1 \ar^{3/2},
\]
as required.

The probability density function of $\widetilde{K}_{\rm cont}$ follows from a general formula
\[
\rho_{\widetilde{K}_{\rm cont}}(k) = \int_{\R}\int_{\R} \rho_{\bm{A}}(\bm{a})\delta\left(k - C_{R_*} a_1 \ar^{3/2}\right)da_1\,d\ar,
\]
where $\delta$ is the Dirac delta. To obtain the result, in the inner integral (in which $\ar$ is treated as fixed), one performs a change of variables from $a_1$ to $\tilde{k} = C_{R_*}a_1\ar^{3/2}$.
\end{proof}

%%%%%%%%%%%%%%%%%%%%%%%%%%%%%%%%%%%%%%%%%%%%%%%%%%%%%%%%%%%%
%%%%%%%%%%%%%%%%%%%%  NEW SECTION   %%%%%%%%%%%%%%%%%%%%%%%%
%%%%%%%%%%%%%%%%%%%%%%%%%%%%%%%%%%%%%%%%%%%%%%%%%%%%%%%%%%%%
\section{Computations}\label{sec:comp}

In this section, we use the stochastic framework developed in Section~\ref{sec:stoch-LJ} to conduct a numerical study of crack propagation. 

%%%%%%%%%%%%%%%%%%%%%%%%%%%%%%%%%%%%%%%%%%%%%%%%%%%%%%%%%%%%
\subsection{Setup}
For our numerical computations, we employ the principles of the recently proposed NCFlex scheme \cite{BK2021}. We fix $\tilde{R}=32$ and consider a computational domain
\begin{equation}
\bm{\Lambda}_{R} = \bm{\La} \cap B_{R},\,\text{ where }\, R = \ell(\tilde{R} + 2R_*),
\end{equation}
then look at displacements $\bm{U}\,\colon\,\bm{\La} \to \R^2$ of the form
\begin{equation}\label{UU}
\bm{U}(\bm{m}) =K\widehat{\bm{u}}(\bm{m} - \bm{\alpha}) + \bm{u}(\bm{m})\,\text{ where } \bm{u}(\bm{m}) = \bm{0}\;\forall\,\bm{m} \in \bm{\La} \,\text{ such that } |\bm{m}| > \ell(\tilde{R}+R_*).
\end{equation}
The rescaling by $\ell$ ensures that, regardless of the choice of $\ell$, for a fixed $R_*$, the computational domain consists of the same number of atoms $N \sim R^2$. The truncation of $\bm{u}$ ensures that the finite-dimensional scheme is consistent with \eqref{Du-Duhat}.

We consider three possible choices for $R_*$, namely:
\begin{enumerate}[label=(\roman*)]
\item $R_* = 1$, which corresponds to accounting for only the nearest neighbor interaction;
\item $R_* = \sqrt{3}$ (second neighbors included too);
\item $R_* = 2$ (up to third neighbors included).
\end{enumerate}

The resulting finite-dimensional approximation to \eqref{en-infinite} is given by
\[
\tilde{\E}(\bm{u},\alpha,K) = \sum_{m \in \bm{\La}_{{R}}} V(\bm{D}\bm{y}^{\bm{U}}(\bm{m})) - V(\bm{D}\bm{y}^{\bm{0}}(\bm{m})),
\]
with $\bm{y}^{\bm{0}}$ as in \eqref{en-infinite}.

The essence of the NCFlex scheme is to employ numerical continuation to trace continuous paths of solutions $(0,1) \ni s \mapsto \bm{U}_s = (\bm{u}_s,\alpha_s,K_s)$, such that 
\begin{equation}\label{ncf-eqn}
(\delta_{\bm{u}} \tilde{\mathcal{E}}(\bm{U}_s),\delta_{\alpha}\tilde{\mathcal{E}}(\bm{U}_s)) = \bm{0}.
\end{equation}
This is a nonlinear system of $2N+1$ equations in $2N+2$ variables and a numerical continuation constraint closes the system.

The specific numerical algorithm employed allows for the quantities of interest to be computed without human supervision. The details are presented in Algorithm~\ref{alg1} and we note that the numerical continuation routine is implemented in Julia using \texttt{BifurcationKit.jl} \cite{BifKit}.

%%%%%%%%%%%%%%%%
\begin{algorithm}[H]
   \caption{Unsupervised NCFlex scheme}\label{alg1}
   \begin{algorithmic}[1]
    \State Given potential parameters $a_1,\ar$ from \eqref{LJ} and some tolerance $\delta$;
      \State Estimate the interval $I = (K_-,K_+)$ by fixing $\bm{u}(\bm{m}) =  0,\;\forall \bm{m} \in \bm{\La}_R$ in \eqref{UU}, and solving $\delta_{\alpha}\tilde{\E}(\bm{U}) = 0$, up to the tolerance $\delta$, for incremental values of $\alpha$ ranging from $-\ell$ to $\ell$ (idea put forward in \cite[Section II.C.2]{BK2021});
      \State Fix $\alpha = \alpha_0$ (e.g., $\alpha_0 = -0.5$), set $K = K_-$ (estimate found in the previous step) and use a Conjugate-Gradient solver with initial guess $(\bm{0},\alpha_0,K_-)$ to find a static boundary equilibrium $\bm{U} = (\bm{u},\alpha_0,K_-)$ satisfying only $\delta_{\bm{u}} \tilde{\E}(\bm{U}) = \bm{0}$, up to tolerance $\delta$ (typically the other equation in \eqref{ncf-eqn} will not be satisfied);
      \State Repeat previous step for incremental values of $K \in (K_-,K_+)$ (estimate found in Step 2) until one identifies $K_0$ for which $\bm{U_0} = (\bm{u_0},\alpha_0,K_0)$ is such that $\delta_{\alpha} \tilde{\E}(\bm{U}) = 0$ (up to tolerance $\delta$) as well, meaning that \eqref{ncf-eqn} holds true (a bisection-type algorithm can be used to speed up the process);
      \State With the first solution $\bm{U_0}$ of the NCFlex scheme identified, apply the numerical continuation routine outlined in \cite[Algorithm 2.]{BK2021} to compute the path of solutions $s \mapsto \bm{U}_s$.
%      \For{$n > 1$}
%  \EndFor
\end{algorithmic}
\end{algorithm}
%%%%%%%%%%%%%%%%

%%%%%%%%%%%%%%%
\begin{figure}[htbp]
\includegraphics[width=0.9\textwidth]{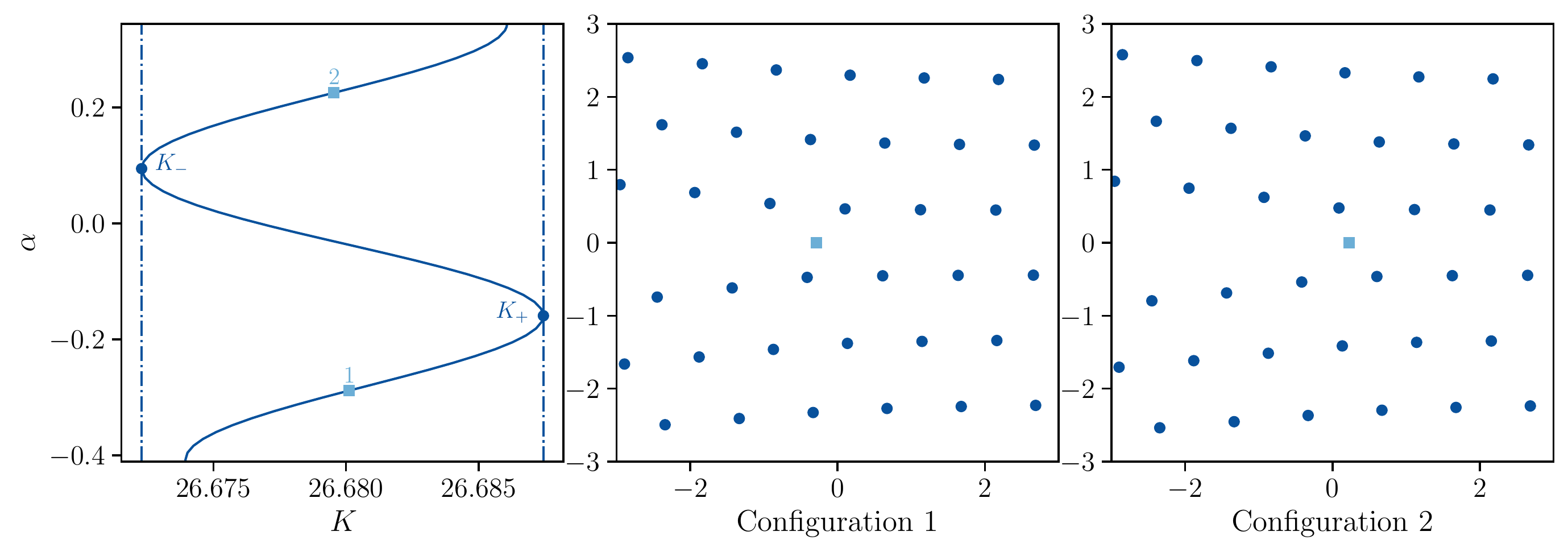}
\caption{Left: The plot $s \mapsto (K_s,\alpha_s)$ showing how $K_+$ and $K_-$ are identified. The leftward-tilt of the curve is a finite-domain effect -- for a theoretical infinite domain the solution curve would be perfectly vertical. Middle and right: atomistic configurations highlighting that a jump by one period on the snaking curve corresponds to the crack propagating by one lattice spacing.}
\label{fig:lat_trap_I}
\end{figure}
%%%%%%%%%%%%%

As noted in Section~\ref{sec:atom-mode1}, the lattice trapping range $I$ and the critical stress intensity factor $\widetilde{K}_{\rm at}$ can be inferred from the computed solution paths (see Figure~\ref{fig:lat_trap_I}). Note, however, that the computed quantities of interest are finite domain approximations. Hence, in particular, $\widetilde{K}^R_{\rm at}$ computed for a domain with radius $R$ will not match the theoretical $\widetilde{K}_{\rm at}$ from Proposition~\ref{prop2}. Therefore, direct comparisons to $\widetilde{K}_{\rm cont}$ are not feasible. Nevertheless, heuristic considerations and numerical evidence point to the fact that
\[
|\widetilde{K}^R_{\rm at} - \widetilde{K}_{\rm at}| \sim \mathcal{O}(R^{-1/2}).
\]

%%%%%%%%%%%%%%%%%%%%%%%%%%%%%%%%%%%%%%%%%%%%%%%%%%%
\subsection{Results} 

We have considered the following cases in our numerical study:
\begin{enumerate}
\item $a_1 = 1$ fixed and a sample of $\num{1000}$ choices of $\ar$ with $\underline{\ar} = 2^{1/6}$ and $\tau = -20$; %and \linebreak (1.2)~${\tau = \num{-400000}}$;
\item $\ar = 2^{1/6}$ fixed and a sample of $\num{1000}$ choices of $a_1$ with $\underline{a_1} = 1$ and $\tau = -20$;% and (2.2)~$\tau = \num{-400000}$;
\item a sample of $\num{1000}$ choices of $\bm{A}$ with $\underline{a_1} = 1$, $\underline{\ar} = 2^{1/6}$ and $\tau = -20$;
\item a combined sample of $\num{1000000}$ of $\bm{A}$ obtained by reusing the samples from (1) and (2);
\item $a_1 = 1$ fixed, a sample of $\num{1000}$ choices of $\ar$ with $\underline{\ar} = 2^{1/6}$ and $\tau = - \num{4000000}$ to test the interplay between the strength of statistical fluctuations and the strength of lattice trapping.
\end{enumerate}
Figure~\ref{fig:stat_fluct} presents the level of statistical fluctuations present in $\phi$ and how this translates to the computed snaking curves. 

%%%%%%%%%%%%%
\begin{figure}[htbp]
\includegraphics[width=0.8\textwidth]{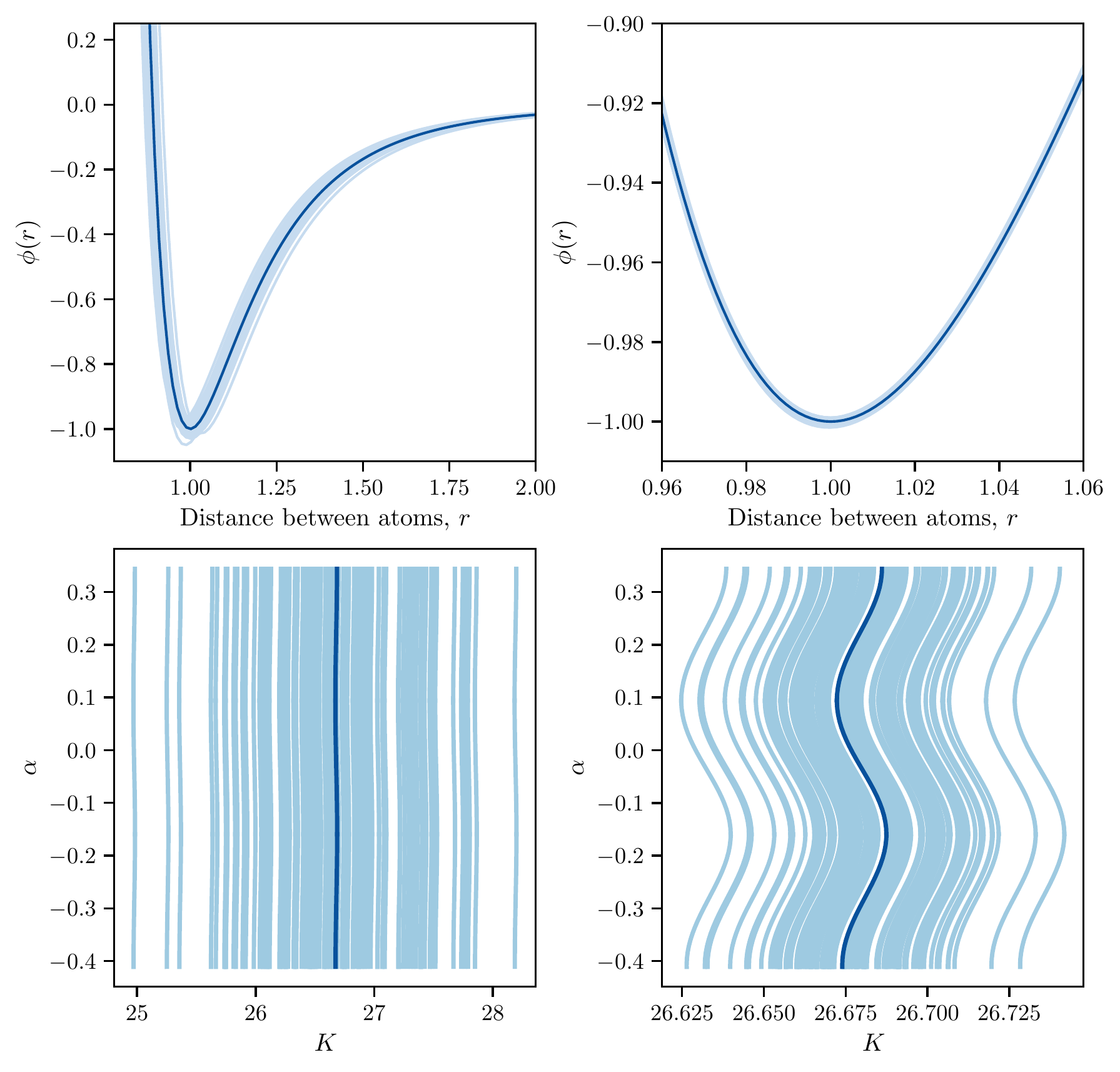}
\caption{Statistical fluctuations for $\tau = -400$ (left column) and $\tau = \num{-4000000}$ (right column). Top row: plot of the interatomic potential $\phi$ for the mean value of parameters (dark blue) and a sample of 100 choices of parameters $a_1,\ar$ (light blue, on the right zoomed-in around unity to actually see the fluctuations). Bottom row: the resulting computed snaking curves.}
\label{fig:stat_fluct}
\end{figure}
%%%%%%%%%%%%%

There are several universal conclusions that can be drawn from our numerical investigation, which we shall now discuss and then refer to in the subsequent subsections detailing each case listed above.

Firstly, it will be numerically verified that the relative strength of the lattice trapping, which we measure as $1 - (K_-/K_+)$, in our model is not a function of $a_1$ or $\ar$, but merely of $R_*$. On a heuristic level, this reflects the fact that the lattice constant $\ell$ is a linear function of $\ar$ and is consistent with the work presented in \cite{curtin_1990}. Our results will also corroborate our conjecture that, in the model considered, $K_+$ and $K_-$, for a fixed domain radius $R$, exhibit the following dependence on $a_1$, $\ar$ and $R_*$,
\begin{equation}\label{KpKm}
K_+ = C^+_{R_*} a_1\,\ar^{3/2}, \quad K_- = C^-_{R_*} a_1\,\ar^{3/2},
\end{equation}
differing from $\widetilde{K}_{\rm cont}$ from Proposition~\ref{prop2} only by a constant which depends on $R_*$. In particular, we will present numerically obtained values for $C^{\pm}_{R_*}$. This is strong evidence that, in fact, the equality $\widetilde{K}_{\rm at} = \widetilde{K}_{\rm cont}$ from \eqref{Kc-Kat} holds true for our model. 

Secondly, the generally nonlinear dependence of quantities of interest on the parameters, as established in Proposition~\ref{prop2} and in \eqref{KpKm}, implies that, e.g., $\mathbb{E}(K_+)$ does not correspond to the deterministic value obtained when parameters are equal to mean values. This alone indicates that employing a purely deterministic approach to model atomistic fracture is of limited practical use. 

Thirdly, the value of the parameter $\tau$ from Proposition~\ref{prop1} plays a crucial role in determining whether the extent of lattice trapping is negligible or not. For $\tau = -20$, it most certainly is, and hence, for this case, since $\widetilde{K}_{\rm at}$ lies somewhere between $K_-$ and $K_+$, we can safely focus on the outer quantities only. However, as $\tau \to -\infty$, lattice trapping starts to dominate over statistical fluctuations. We show this by considering the extreme case with $\tau = \num{-4000000}$.

We now present the results of our numerical study.

%%%%%%%%%%%%%%
\begin{figure}[htbp]
	\includegraphics[width=0.8\textwidth]{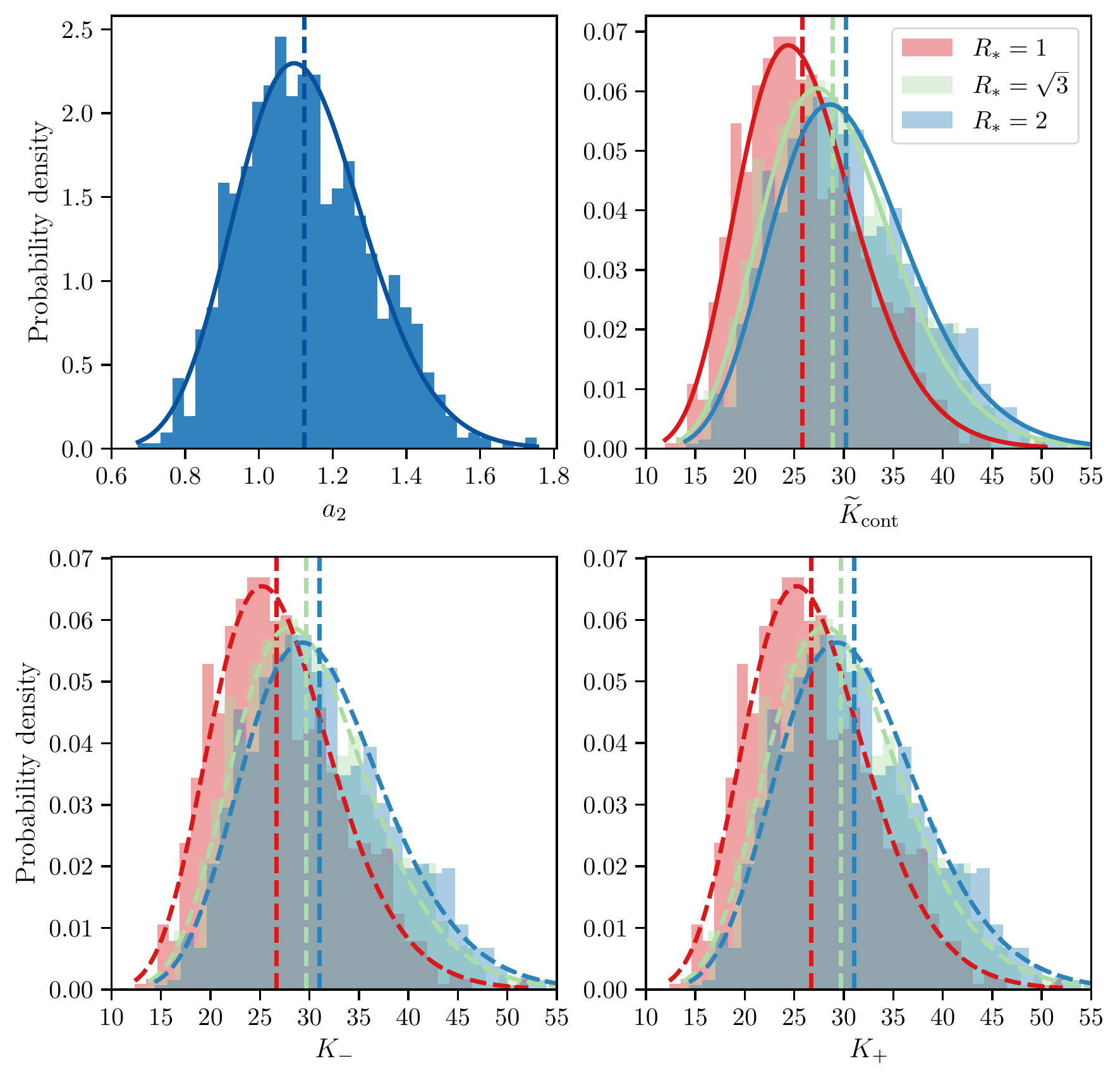}
	\caption{Case (1). Top left: the normalized histogram of a sample of 1000 choices of $\ar$ drawn from the MaxEnt probability distribution established in Proposition~\ref{prop1}, with $\underline{a_1} = 2^{1/6}$ and $\tau = -20$, together with the probability density function $\rho_{A_2}$. Top right: the histogram and probability density function for $\widetilde{K}_{\rm cont}$ from Proposition~\ref{prop2} for $R_* = 1,\,\sqrt{3},\,2$. Bottom: the resulting numerically computed histogram of values of $K_-$ (left) and $K_+$ (right). The dotted lines are numerically predicted probability density functions, based on \eqref{KpKm}, with values of $C^+_{R_*}, C^-_{R_*}$ reported in Table~\ref{tab:a2}.}\label{fig:a2_hist}
\end{figure}
%%%%%%%%%%%%%%%%

%%%%%%%%%%%%%%%
\begin{figure}[htbp]
	\includegraphics[width=0.8\textwidth]{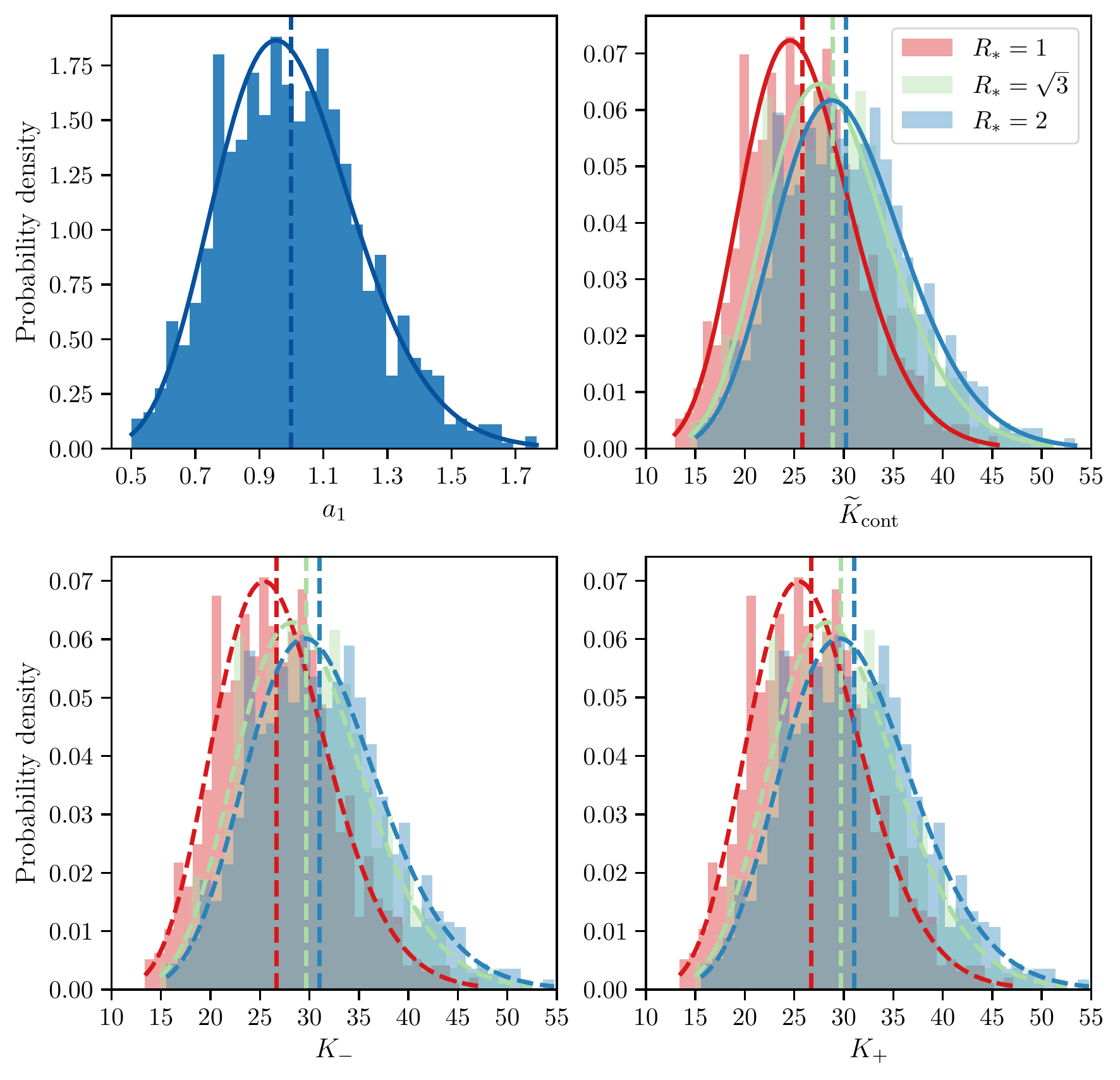}
	\caption{Case (2). Top left: the normalized histogram of a sample of 1000 choices of $a_1$ drawn from the MaxEnt probability distribution established in Proposition~\ref{prop1}, with $\underline{a_1} = 1$ and $\tau = -20$, together with the probability density function $\rho_{A_1}$. Top right: the histogram and probability density function for $\widetilde{K}_{\rm cont}$ from Proposition~\ref{prop2} for $R_* = 1,\,\sqrt{3},\,2$. Bottom: the resulting numerically computed histogram of values of $K_-$ (left) and $K_+$ (right). The dotted lines are numerically predicted probability density functions, based on \eqref{KpKm}, with values of $C^+_{R_*}, C^-_{R_*}$ reported in Table~\ref{tab:a1}.}\label{fig:a1_hist}
\end{figure}
%%%%%%%%%%%%%%%%

%%%%%%%%%%%%%
\begin{figure}[htbp]
	\includegraphics[draft=false,width=0.95\textwidth]{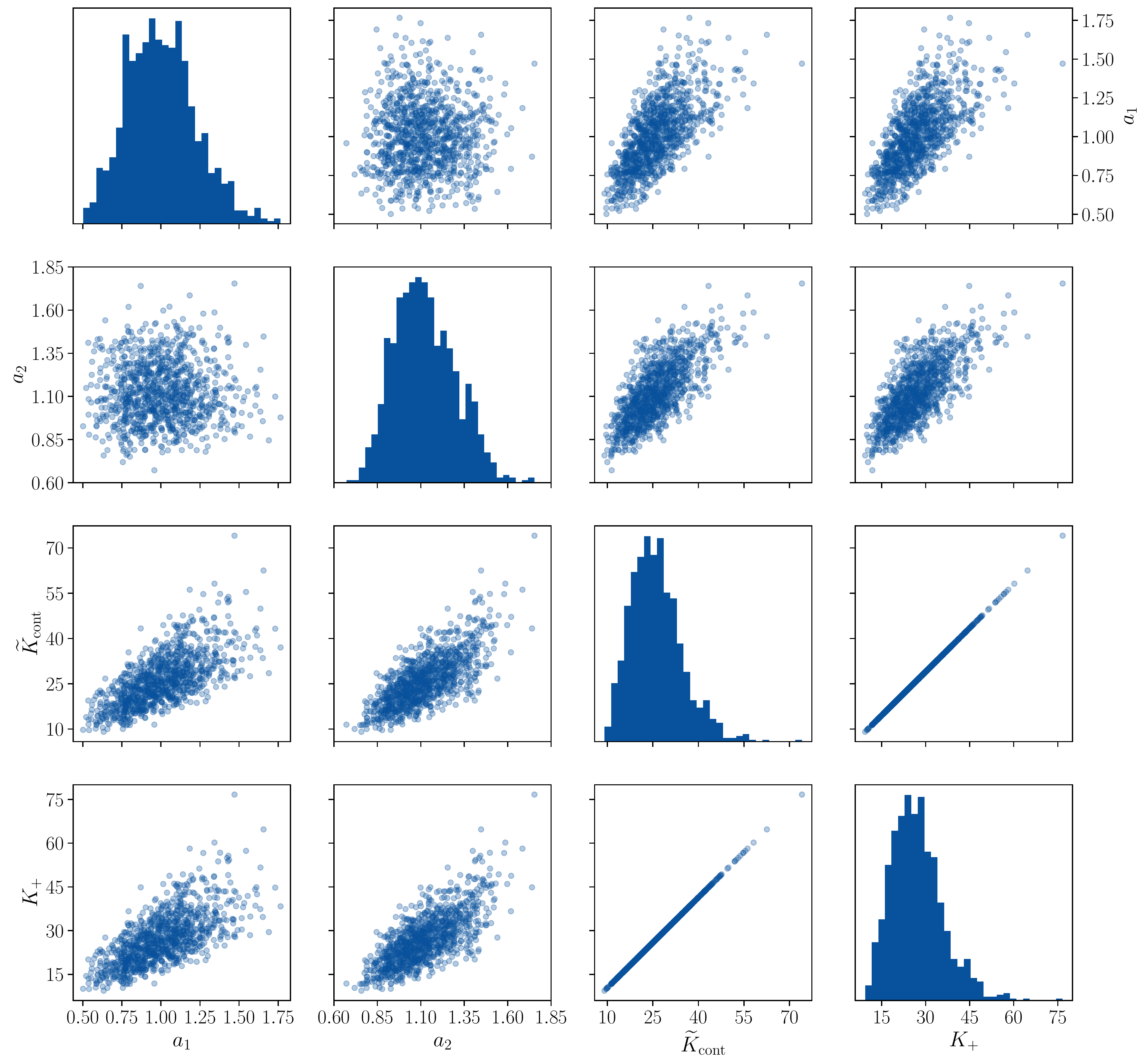}
	\caption{Case (3) when $R_{\ast} = 1$: a scatter matrix plot for variables $a_1,\,\ar,\,\widetilde{K}_{\rm cont}, K_+$, highlighting bivariate relationships between them (off the diagonal). On the diagonal histograms of the variables are presented (note that for histograms the vertical axis does not apply). The plots in particular confirm that $K_+$ is determined by $\widetilde{K}_{\rm cont}$ (and the fixed parameters $\tilde{R}$ and $R_{\ast}$).}
	\label{fig:scatmat_N1}
\end{figure}
%%%%%%%%%%%

%%%%%%%%%%%%%%%%
\begin{figure}[htbp]
	\includegraphics[trim=0 0 0 0pt, clip,width=0.8\textwidth]{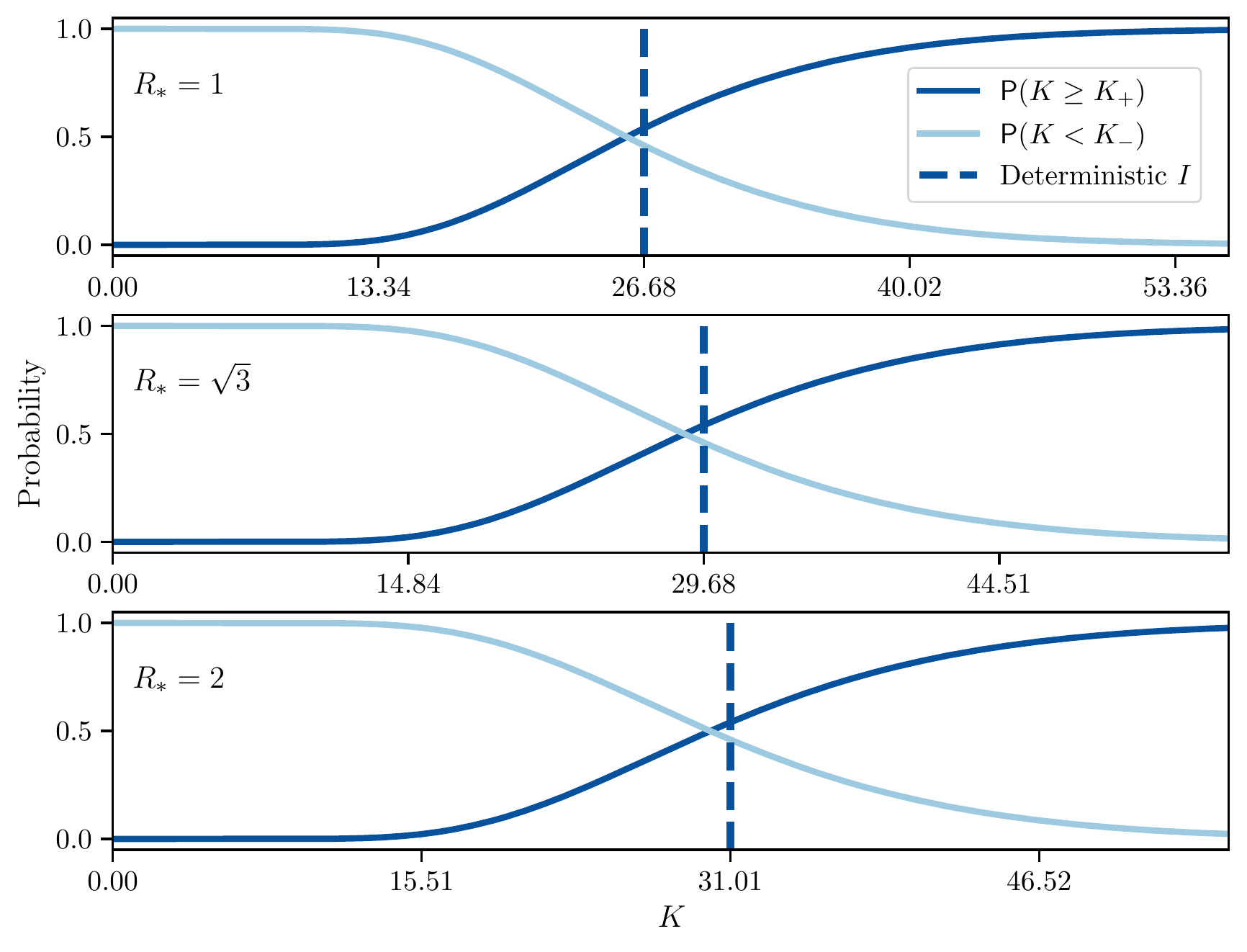}
	\caption{Case (4) when $R_{\ast} = 1,\,\sqrt{3},\,2$: the probability of the crack propagating at a given $K$, computed as $\mathsf{P}(K \geq K_+)$ and of not propagating, computed $\mathsf{P}(K<K_-)$ compared with the deterministic interval $I = (K_-,K_+)$ computed for the mean values  $\underline{a_1}$ and $\underline{\ar}$. Note that when $\tau = -20$, the lattice trapping strength is negligible compared to the statistical fluctuations and $I$ can be, effectively, treated as a single value. We further note that the probabilities were computed both analytically (using \eqref{KpKm} and Proposition~\ref{prop2}), as well as from data and for a sample of this size they are indistinguishable.}\label{fig:prob_case4}
\end{figure}
%%%%%%%%%%%%%%%

%%%%%%%%%%%%%%%
\begin{figure}[htbp]
	\includegraphics[trim=0 0 0 0pt, clip,width=0.8\textwidth]{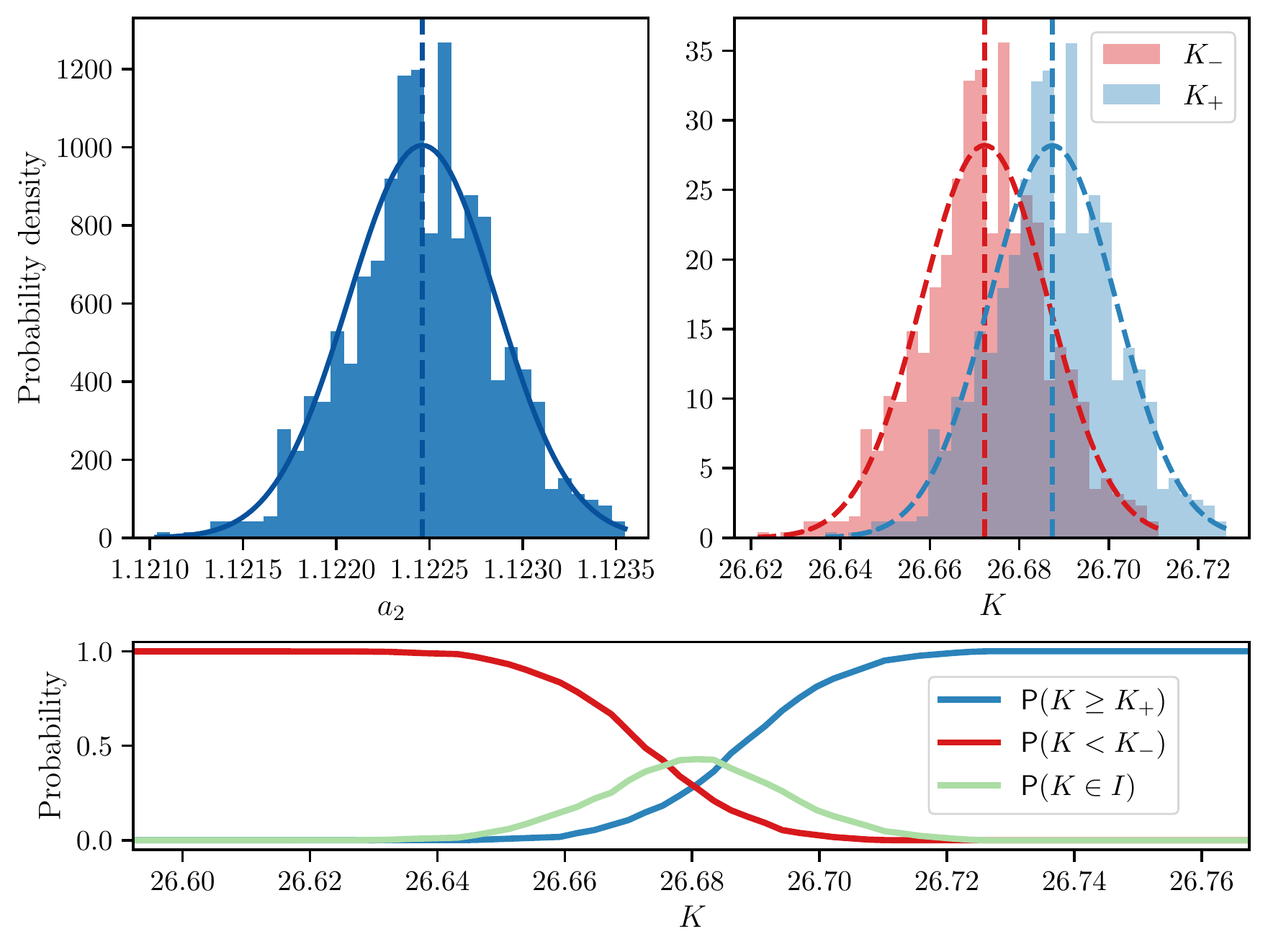}
	\caption{Case (5) when $R_{\ast} = 1$. Top left: the normalized histogram of a sample of 1000 choices of $\ar$ drawn from the MaxEnt probability distribution established in Proposition~\ref{prop1}, with $\underline{a_1} = 2^{1/6}$ and $\tau = \num{-4000000}$, together with the probability density function $\rho_{A_2}$. Top right: the histogram and probability density function for $K_-$ and $K_+$, with dotted lines are numerically predicted probability density functions, based on \eqref{KpKm}. Bottom: the probability of crack propagation at a given $K$, as in Figure~\ref{fig:prob_case4}. }\label{fig:case5}
\end{figure}
%%%%%%%%%%%%%%%

%%%%%%%%%%%%%
\subsubsection*{Case (1): $a_1 = 1$ fixed, $\ar$ sampled with $\tau = -20$.} 

We first consider the case where $a_1 = 1$ remains fixed and the parameter $\ar$ is sampled from the MaxEnt probability distribution established in Proposition~\ref{prop1}, with $\underline{\ar} = 2^{1/6}$ (corresponding to the lattice constant $\ell = 1$ when $R_{\ast} = 1$) and $\tau = -20$. The sample is $\{\ar^{(i)}\}_{i=1}^S$ where $S=1000$. The probability density function (pdf) from which the sample was drawn, and the histogram of the sample are presented in Figure~\ref{fig:a2_hist}. In this figure, we also present the quantities of interest from \eqref{QoI} for $R_{\ast} = 1,\,\sqrt{3},\,2$, that is, the pdf and the histogram of $\{\widetilde{K}^{(i)}_{\rm cont}\}$ computed via Proposition~\ref{prop2}, and the histograms of $\{K_{\pm}^{(i)}\}$, with a pdf fitted according to \eqref{KpKm}. Table~\ref{tab:a2} complements the analysis by gathering the relevant data. In particular, we report that the relative strength of the lattice constant only varies with $R_*$ and is rather small, varying from just $0.05\%$ for $R_* = 1$ to $0.07\%$ for $R_*=2$. The data in Table~\ref{tab:a2} confirms that $\mathbb{E}(K_+)$ does not equal the deterministic $K_+$ computed for the mean value of $\ar$ (the same applies to $K_-$ and $\widetilde{K}_{\rm cont})$. We also report on the numerically computed values for $C_{R_*}^+$ and $C_{R_*}^-$ from \eqref{KpKm} and how they compare with $C_{R_*}$, which can be obtained analytically based on the proof of Proposition~\ref{prop2}. 

%%%%%%%%%%%%%%
\begin{table}[htbp]
    \caption{Case (1), relevant data as $R_*$ varies: the relative strength of the lattice trapping measured as $1 - K_+/K_-$; the expected value of $K_+$; the deterministic value of $K_+$ and constants $C_{R_*},C^+_{R_*},C^-_{R_*}$ from Proposition~\ref{prop2} and \eqref{KpKm}.}
    \centering
    \begin{tabular}{c|c|c|c|c|c|c}
      $R_{\ast}$ & $1 - K_-/K_+$ & $\mathbb{E}(K_+)$ & $K_+$ at $\underline{\ar}$ & $C_{R_*}$ &  $C^-_{R_*}$  & $C_{R_*}^+$ \\
         \hline
     $1$ & 0.0005676 & 26.9643 & 26.6874 & 21.6864 & 22.4286 & 22.4414  \\
     $\sqrt{3}$ & 0.0006815 & 29.9944 & 29.6865 & 24.2825 & 24.9462 &  24.9632  \\
     $2$ & 0.0007081 & 31.3467 & 31.0249 & 25.4237 & 26.0702 & 26.0887\vspace{5pt}\\ 
     \end{tabular}     
    \label{tab:a2}
\end{table}
%%%%%%%%%%%%%%

%%%%%%%%%%%%%%%%%%%%%
\subsubsection*{Case (2): $\ar = 2^{1/6}$ fixed, $a_1$ sampled with $\tau = \num{-20}$.} 

Next, we consider the case where $\ar = 2^{1/6}$ remains fixed and the parameter $a_1$ is sampled from the MaxEnt probability distribution established in Proposition~\ref{prop1}, with $\underline{a_1} = 1$ and $\tau = -20$. The sample is $\{a_1^{(i)}\}_{i=1}^S$ where $S=1000$. Figure~\ref{fig:a1_hist} and Table~\ref{tab:a1} summarize our findings for this case. We note that these results can be obtained very quickly, as the NCFlex scheme only has to be run once due to the following remark.
\begin{remark}
Assume that $\bm{U} = (\bm{u},\alpha,K)$ specifies an equilibrium configuration
\[
\bm{y}^{\bm{U}}(\bm{m}) = \bm{m} + K\widehat{\bm{u}}(\bm{m}-\bm{\alpha}) + \bm{u}(\bm{m})
\]
which solves \eqref{ncf-eqn} for some choice of the parameters $a_1$ and $\ar$ from the interatomic potential \eqref{LJ}. It follows from \eqref{u_CLE} that a multiplicative inverse of the shear modulus $\mu$ enters as a prefactor in $\widehat{\bm{u}}$, whereas from \eqref{mu-a1a2} it follows that the shear modulus $\mu$ depends on $a_1$ linearly. In a pointwise sense, the equilibrium $\bm{y}^{\bm{U}}$ satisfies, for each $\bm{m} \in \bm{\La}_R$,
\[
\sum_{\bm{\rho} \in \Rc} \Big[\phi'\big(|D_{\bm{\rho}}\bm{y}^{\bm{U}}(\bm{m}-\bm{\rho})|\big) - \phi'\big(|D_{\bm{\rho}}\bm{y}^{\bm{U}}(\bm{m})|\big) \Big] = 0
\] 
and since $a_1$ enters as a prefactor in $\phi'$, it readily follows that 
\begin{equation}\label{Utilde}
\bm{\widetilde{U}} = (\bm{u},\alpha,K\,a_1 / \widetilde{a_1})
\end{equation}
specifies an equilibrium configuration for the model in which the first parameter in the interatomic potential from \eqref{LJ} is set to $\widetilde{a}_1$. As a result, a snaking curve $s \mapsto \bm{U}_s$ obtained by running the NCFlex scheme for one value of $a_1$ gives rise to the corresponding snaking curve $s \mapsto \bm{\widetilde{U}}_s$ via the transformation in  \eqref{Utilde}.
\end{remark}
This observation implies that working with the 2D random variable $\bm{A}$ defined in \eqref{A-ran} is only as computationally costly as working with $\ar$, so we proceed to Cases 3 \& 4. 

%%%%%%%%%%%%%%%%
\begin{table}[htbp]
    \caption{Case (2), relevant data as $R_*$ varies: the relative strength of the lattice trapping measured as $1 - K_+/K_-$; the expected value of $K_+$; the deterministic value of $K_+$ and constants $C_{R_*},C^+_{R_*},C^-_{R_*}$ from Proposition~\ref{prop2} and \eqref{KpKm}. The only differences compared to the data for Case 1.1 presented in Table~\ref{tab:a2} are highlighted in bold.}
    \centering
    \begin{tabular}{c|c|c|c|c|c|c}
      $R_{\ast}$ & $1 - K_-/K_+$ & $\mathbb{E}(K_+)$ & $K_+$ at $\underline{a_1}$ & $C_{R_*}$ &  $C^-_{R_*}$  & $C_{R_*}^+$ \\
         \hline
     $1$ & 0.0005676 & \textbf{26.8146} & 26.6874 & 21.6864 &  22.4286 & 22.4414  \\
     $\sqrt{3}$ & 0.0006815 & \textbf{29.8280} & 29.6865 & 24.2825 & 24.9462 &  24.9632  \\
     $2$ & 0.0007081 & \textbf{31.1727} & 31.0249 & 25.4237 & 26.0702 & 26.0887\vspace{5pt}\\ 
     \end{tabular}     
    \label{tab:a1}
\end{table}
%%%%%%%%%%%%%%%

%%%%%%%%%%%%%%%%%%
\subsubsection*{Case (3): $a_1$ and $\ar$ sampled with $\tau = -20$.}

We now consider the case where both $a_1$ and $\ar$ are sampled simultaneously from the MaxEnt probability distribution established in Proposition~\ref{prop1}, with $\underline{a_1} = 1$, $\underline{\ar} = 2^{1/6}$ (corresponding to the lattice constant $\ell = 1$ when $R_{\ast} = 1$) and $\tau = -20$. In particular, the sample is $\{(\tilde{a}_1^{(i)},\tilde{a}_2^{(i)})\}_{i=1}^{S}$ where $S=1000$. We present the resulting data in the form a scatter matrix plot to emphasize the bivariate dependence between the random variables involved. This is shown in Figure~\ref{fig:scatmat_N1} for the case when $R_* = 1$. The perfect linear dependence between $\widetilde{K}_{\rm cont}$ and $K_+$ provides further numerical evidence that, in fact, \eqref{KpKm} holds true, rendering the ratio $\widetilde{K}_{\rm cont} / K_+$ a function of $R_*$ only (for a fixed $R$). This again strongly hints at the veracity of \eqref{Kc-Kat}. We further see the statistical independence of $a_1$ and $\ar$ (by design) and the qualitatively different dependence of the quantities of interest on $a_1$ and $\ar$. 

%%%%%%%%%%%%%%%%%%%
\subsubsection*{Case (4): combining samples of $a_1$ and $\ar$ when $\tau = -20$.}

In this case, we take the samples $\{a_1^{(i)}\}$ from Case 2 and $\{\ar^{(i)}\}$ from Case 1 into a combined $\num{1000000}$ sample $\{(a_1^{(i)},\ar^{(j)})\}_{i,j=1}^{S}$, where, as before, $S=1000$. This is made easy by the observation in Case 2, which implies that the NCFlex scheme only has to be run a 1000 times and not a \num{1000000} times. In particular, our focus is on the probability of a crack propagating or not propagating. Due to the phenomenon of lattice trapping, one can distinguish three possibilities: 
\begin{enumerate}[label=({\Alph*})]
\item  if $K < K_-$ then the crack will  definitely not propagate;
\item  if $K_- \leq K < K_+$ (in other words, $K \in I$) then the crack remains lattice-trapped; 
\item if $K \geq K_+$ then the crack will definitely propagate. 
\end{enumerate}
In the lattice-trapped case thermal fluctuations typically present at temperature above the absolute zero imply there is a non-zero probability of the crack propagating. This is a highly non-trivial case, which we do not delve into, but note that such questions can be approached by combining our approach with the framework of transition state theory \cite{HTB90}.  The key quantity here is the energy barrier at different values of $K$ within the lattice trapping range, which can be achieved with the NCFlex scheme. 
In our stochastic framework, (A) can be restated as $\mathsf{P}(K<K_-)$, (B) as $\mathsf{P}(K \in I$ and (C) as $\mathsf{P}(K \geq K_+)$. At $\tau = -20$, case (B) is negligible, hence we omit it from plots and only show (A) and (C), both obtained analytically and from the data in Figure~\ref{fig:prob_case4}.

%%%%%%%%%%%%%%%%%%%
\subsubsection*{Case (5): as in Case (1) but with $\tau = \num{-4000000}$}

In the final case, we revisit the setup from Case (1), but adjust the statistical fluctuations parameter to $\tau = \num{-4000000}$. In this case, the support of the probability density function is heavily concentrated around the mean, to the point where the strength of the lattice trapping is comparable with statistical fluctuations. This implies that case (B) discussed in Case (4) ceases to be negligible. As seen from Figure~\ref{fig:case5}, at this level of statistical fluctuations, there is a significant shift between the pdfs of $K_-$ and $K_+$. As a result, for the values of $K$ within the lattice trapping range, $\mathsf{P}(K<K_-)$ and $\mathsf{P}(K\geq K_+)$ are not complementary, in the sense that  they do not add up to approximately $1$, as can be seen by the considerable probability of $\mathsf{P}(K \in I)$ in-between the mean values of $K_-$ and $K_+$. This confirms that, in this the case, the strength of the lattice trapping begins to dominate over the strength of statistical fluctuations. This effect can be far more pronounced already at more reasonable values of $\tau$ in other models where the lattice trapping range is not as small as in our case.

This concludes our numerical investigation, in which we explored an implementation of the stochastic framework introduced in Section~\ref{sec:stoch-LJ}.

%%%%%%%%%%%%%%%%%%%%%%%%%%%%%%%%%%%%%%%%%%%%%%%%%%%%%%%%%%%%
%%%%%%%%%%%%%%%%%%%%  NEW SECTION   %%%%%%%%%%%%%%%%%%%%%%%%
%%%%%%%%%%%%%%%%%%%%%%%%%%%%%%%%%%%%%%%%%%%%%%%%%%%%%%%%%%%%
\section{Conclusion}

We have introduced an information-theoretic stochastic framework for studying atomistic crack propagation in the analytically-tractable case of the so-called theoretical Lennard-Jonesium 2D solid with the ground state of a triangular lattice and undergoing a pure Mode I fracture. In particular, we invoked the Maximum Entropy Principle to argue that, when little information is available, except for the mean values of the parameters, the parameters in the Lennard-Jones potential should be modeled as independent, Gamma-distributed random variables. Due to the relative simplicity of the model, we were able to infer how the uncertainty in the choice of these parameters propagate to quantities of interest, which in the case of atomistic fracture is the range of lattice trapping and the value of the critical stress intensity factor. This was followed by an extensive numerical study of stochastic atomistic fracture, made possible by an automated formulation of the NCFlex scheme from \cite{BK2021}, which, in particular, highlighted the limitations of a purely deterministic approach. In future work, we aim to develop a more general information-theoretic approach to uncertainty quantification in atomistic material modeling,  and further explore the stochastic effects within the lattice trapping range. \\

%%%%%%%%%%%%%%%%%%%%%%%%%%%%%%%%%%%
\paragraph{Acknowledgement.} The support by the Engineering and Physical Sciences Research Council of Great Britain under research grant EP/S028870/1 to Maciej Buze and L. Angela Mihai is gratefully acknowledged.

%\paragraph{Conflict of interest statement.} The authors declare that they have no conflict of interest.  

%\paragraph{Data availability statement.} There are no additional data associated with this paper.

%%%%%%%%%%%%%%%%%%%%%%%%%%%%%
%%%%%%%%%%%%%%%%%%%%%%%%%%%%%%%%%%%%%%%%%%%%%%%%%%%%%%%%%%%%
%%%%%%%%%%%%%%%%%%%%  APPENDIX   %%%%%%%%%%%%%%%%%%%%%%%%
%%%%%%%%%%%%%%%%%%%%%%%%%%%%%%%%%%%%%%%%%%%%%%%%%%%%%%%%%%%%
\appendix

\section{Determining the lattice constant, the shear modulus and the surface energy}\label{appendA}

In this appendix, we present calculations confirming the veracity of the formulae given by \eqref{lat-const-dep} and \eqref{mu-a1a2}. Such calculations are well known in the literature, but worth elaborating upon since they are central to our stochastic framework. We start with the formally defined energy
\[
\mathcal{E}(\bm{U}) = \sum_{\bm{m} \in \bm{\La}_R} V(\bm{D}\bm{y}^{\bm{U}}(\bm{m})),
\]
where $\bm{\La}_R = \bm{\La} \cap B_R$. We recall that $\bm{U}$ is the displacement and $\bm{y}^{\bm{U}}(\bm{m}) = \bm{m} + \bm{U}(\bm{m})$ is the deformation. A formal Taylor expansion of this energy around $\bm{y}^{\bm{0}}$ to second order yields
\[
\E(\bm{y}^{\bm{U}}) = \E(\bm{y}^{\bm{0}}) + \<{\delta \E(\bm{y}^{\bm{0}})}{\bm{U}} + \<{\delta^2 \E(\bm{y}^{\bm{0}})\bm{U}}{\bm{U}} + {\rm h.o.t.},
\]
where
\[
\<{\delta \E(\bm{y}^{\bm{0}})}{\bm{U}} = \sum_{\bm{m} \in \bm{\La}_R} \nabla V((\bm{\rho})) : \bm{D}\bm{U}(\bm{m}) = \sum_{\bm{m} \in \bm{\La}_R} \sum_{i, \bm{\rho}} \partial_{i\bm{\rho}}V((\bm{\rho}))D_{\bm{\rho}}U_i(\bm{m}).
\]
and 
\begin{align*}
\<{\delta^2 \E(\bm{y}^{\bm{0}})\bm{U}}{\bm{U}} &= \sum_{\bm{m} \in \bm{\La}_R} \nabla^2V((\bm{\rho}))\bm{D}\bm{U}(\bm{m}) : \bm{D}\bm{U}(\bm{m})\\
 &= \sum_{\bm{m} \in \bm{\La}_R} \sum_{i,\bm{\rho},j,\bm{\sigma}} \partial^2_{i\bm{\rho} j \bm{\sigma}}V((\bm{\rho}))D_{\bm{\rho}}U_i(\bm{m})D_{\bm{\sigma}}U_j(\bm{m}).
\end{align*}
For a uniform displacement $\bm{U}$, of the form $\bm{U}(\bm{x}) = \bm{F}\bm{x}$, for some suitable $\bm{F} \in \R^{2\times 2}$, we have $D_{\bm{\rho}}\bm{U}(\bm{x}) = \nabla \bm{U}(\bm{x}) \bm{\rho}$. This implies that for uniform displacements
\begin{equation}\label{deltaE0}
\<{\delta\E(\bm{y}^{\bm{0}})}{\bm{U}} = \sum_{\bm{m} \in \bm{\La}_R} \sum_{i,\alpha=1}^2 L_{i\alpha}\partial_{\alpha} U_i,
\end{equation}
where, due to the form of the potential, we have 
\begin{equation}\label{Lij}
L_{i\alpha} = \sum_{\bm{\rho}\in{\Rc}} \frac{\phi'(|\bm{\rho}|)}{|\bm{\rho}|}\rho_i\rho_{\alpha}.
\end{equation}
It is natural to assume that the potential in place admits the perfect lattice as an equilibrium configuration, and for that to be the case, it is necessary that $\<{\delta\E(\bm{y}^{\bm{0}})}{\bm{U}} = \bm{0}$, for any uniform displacement $\bm{U}$. It follows that the potential parameters $a_1$ and $\ar$ in \eqref{LJ}, and the lattice constant $l$ have to be chosen so that
\begin{equation}\label{cond1}
\sum_{i,\alpha=1}^2 L_{i\alpha}\partial_{\alpha} U_i = 0,
\end{equation}
for any $\bm{U}$. A direct calculation reveals that
\[
L_{i\alpha} = 24a_1\ar^{-6} \sum_{\bm{\rho} \in \Rc} \rho_i\rho_{\alpha}|\bm{\rho}|^{-14}\left(|\bm{\rho}|^6 - 2\ar^{-6}\right).
\]
Due to the lattice symmetries in the interaction range ${\Rc}$, it is immediate that, for $i,\alpha \in \{1,2\}$, $L_{i\alpha} = \hat{L}\delta_{i\alpha}$, where $\delta_{ij}$ denotes the Kronecker delta and
\[
\hat{L}=24\frac{a_1 \ar^{-6}}{\ell^{12}}\left(A_{R_*} \ell^6 - B_{R_*} \ar^{-6}\right).
\]
The constants depending on $R_*$ are 
\[
A_{R_*} = \sum_{\bm{\hat{\rho}} \in \hat{\Rc}}\hat{\rho}_1^2|\bm{\hat{\rho}}|^{-8}\,\quad B_{R_*} = \sum_{\bm{\hat{\rho}} \in \hat{\Rc}}2\hat{\rho}_1^2|\bm{\hat{\rho}}|^{-14},
\]
where $\hat{\Rc} = \Rc / \ell$ (i.e., with lattice constant normalized to unity). It follows that the lattice constant $\ell$ is a function of $R_*$ and $a_2$, since 
\begin{equation}\label{LJsigma-a-rel}
A_{i\alpha} = 0\quad\forall\, i,\alpha \implies \ell = \left(\frac{B_{R_*}}{A_{R_*}}\right)^{1/6} \ar^{-1}.
\end{equation}
A similar line of reasoning can be used to establish \eqref{mu-a1a2}. The lattice symmetries present in ${\Rc}$ imply that the only non-zero entries of the associated elasticity tensor $\C$ from \eqref{C-pp} are $\C_{iiii}$ ($i=1,2$) and $\C_{iijj} = \C_{ijij}$ ($i=1,2$ and $j=1,2$, $j\neq i$), and, in fact,
\begin{align*}
\C_{iiii} &= \frac{1}{{\rm det}(\ell\bm{M})}a_1 \sum_{\bm{\hat{\rho}} \in \hat{\Rc}}\hat{\rho}_1^4\left((B_1+B_3)|\bm{\hat{\rho}}|^{-16} - (B_2+B_4)|\bm{\hat{\rho}}|^{-10}\right),\\
\C_{iijj} &= \frac{1}{{\rm det}(\ell\bm{M})}a_1 \sum_{\bm{\hat{\rho}} \in \hat{\Rc}}\hat{\rho}_1^2\hat{\rho}_2^2\left((B_1+B_3)|\bm{\hat{\rho}}|^{-16} - (B_2+B_4)|\bm{\hat{\rho}}|^{-10}\right),
\end{align*}
for known constants $B_1,\dots,B_4$ depending only on $R_*$. As a result, we have the shear modulus given by
\begin{equation}\label{mu2}
\mu = \frac{1}{3}\C_{1111} = \C_{1122} = \C_{1212} = D_{R_*} a_1 \ar^{2},
\end{equation}
where the dependence on $\ar$ enters through \eqref{LJsigma-a-rel}. Finally, we also show the surface energy computation that confirms \eqref{gamma-gen}. Let $d_n$ denote the distance to the $n$th neighbor in the triangular lattice, with lattice constant equal to unity, and let $N_*$ be the unique value such that 
\[
d_{N_*} \leq R_*\;\text{ but }\;d_{N_* +1} > R_*.
\]
For instance, if $R_* = 2$ then $N_* = 3$, since $d_1 = 1$, $d_2 = \sqrt{3}$ and $d_3 = 2$. If the crack surface is extended by $L$, then $m_n L /\ell$ bonds of length $d_n$ additionally cross from one side of the crack to the other. For instance, $m_1 = 2$ (two nearest-neighbor bonds cross the surface in the triangular lattice if we extend the surface by one lattice spacing) and $m_2 = 4$. Importantly, $m_n$ is a fixed constant for each $n$. The energetic cost of breaking each such bond is given by $-\phi(d_n \ell)$. In the light of the above, a general formula for the surface energy $\gamma$ can be stated as
\[
\gamma = \frac{1}{L} \frac{L}{\ell} \left( \sum_{n = 1}^{N_*} - m_n \phi(d_n \ell)\right).
\]
For a general $R_*$, the relationship between $\ell$ and $\ar$ established in \eqref{lat-const-dep} implies that, for any scalar $\alpha$, we have 
\[
\phi(\alpha \ell) = 4a_1 \ar^{-6}\alpha^{-6} \frac{A_{R_*}}{B_{R_*}}\ar^6\left(\ar^{-6}\alpha^{-6} \frac{A_{R_*}}{B_{R_*}}\ar^6 - 1\right) =: C a_1,
\]
where the constant $C$ depends only on $\alpha$ and $R_*$, as the terms involving $\ar$ cancel one another out. Putting it all together, it is only a matter of gathering all the different constants depending only on $R_*$ to conclude that 
\[
\gamma = \tilde{C}_{R_*} a_1 \ar,
\]
where the constant $\tilde{C}_{R_*}$ only depends on $R_*$.

%%%%%%%%%%%%%%%%%%%%%%%%%%%%%%%%%%%%%%%%%%%%%%%%%%%%%%%%%%%%
%%%%%%%%%%%%%%%%%%%%  REFERENCES  %%%%%%%%%%%%%%%%%%%%%%%%
%%%%%%%%%%%%%%%%%%%%%%%%%%%%%%%%%%%%%%%%%%%%%%%%%%%%%%%%%%%%
\bibliographystyle{amsplain}

%\bibliography{papers}

\newcommand{\noop}[1]{}
\providecommand{\bysame}{\leavevmode\hbox to3em{\hrulefill}\thinspace}
\providecommand{\MR}{\relax\ifhmode\unskip\space\fi MR }
\providecommand{\MRhref}[2]{%
  \href{http://www.ams.org/mathscinet-getitem?mr=#1}{#2}
}
\providecommand{\href}[2]{#2}
\begin{thebibliography}{10}

\bibitem{Bitzek2015}
E.~Bitzek, J.~R. Kermode, and P.~Gumbsch, \emph{Atomistic aspects of fracture},
  International Journal of Fracture \textbf{191} (2015), no.~1, 13--30.

\bibitem{2017-bcscrew}
J.~Braun, M.~Buze, and C.~Ortner, \emph{The effect of crystal symmetries on the
  locality of screw dislocation cores}, SIAM J. Math. Anal. \textbf{51} (2019),
  no.~2, 1108--1136.

%\bibitem{BDO18}
%J.~Braun, M.~H. Duong, and C.~Ortner, \emph{Thermodynamic limit of the
%  transition rate of a crystalline defect}, Arch.\ Ration. Mech.\ Anal.
%  \textbf{238} (2020), 1413--1474.

\bibitem{2018-antiplanecrack}
M.~Buze, T.~Hudson, and C.~Ortner, \emph{Analysis of an atomistic model for
  anti-plane fracture}, Mathematical Models and Methods in Applied Sciences
  \textbf{29} (2019), no.~13, 2469--2521.

\bibitem{2019-antiplanecrack}
\bysame, \emph{Analysis of cell size effects in atomistic crack propagation},
  {ESAIM}: Mathematical Modelling and Numerical Analysis \textbf{54} (2020),
  no.~6, 1821--1847.

\bibitem{BK2021}
M.~Buze and J.R. Kermode, \emph{Numerical-continuation-enhanced flexible
  boundary condition scheme applied to mode-i and mode-iii fracture}, Phys.
  Rev. E \textbf{103} (2021), 033002.

\bibitem{curtin_1990}
W.~A. Curtin, \emph{On lattice trapping of cracks}, Journal of Materials
  Research \textbf{5} (1990), no.~7, 1549–1560.

\bibitem{emingstatic}
W. E and P. Ming, \emph{{Cauchy}-{Born} rule and the stability of
  crystalline solids: Static problems}, Arch.\ Ration. Mech.\ Anal.
  \textbf{183} (2007), 241--297.

\bibitem{Ericksen08}
J.~L. Ericksen, \emph{On the {Cauchy}-{Born} rule}, Mathematics and Mechanics
  of Solids \textbf{13} (2008), no.~3-4, 199--220.

\bibitem{E84}
J.L Ericksen, \emph{The {Cauchy} and {Born} hypotheses for crystals}, Phase
  Transformations and Material Instabilities in Solids (ed. M.E. Gurtin)
  (1984), 61--77.

\bibitem{Frederiksen2004}
S.~L. Frederiksen, K.~W. Jacobsen, K.~S. Brown, and J.~P.
  Sethna, \emph{Bayesian ensemble approach to error estimation of interatomic
  potentials}, Phys. Rev. Lett. \textbf{93} (2004), 165501.

\bibitem{friesecketheil}
G.~Friesecke and F.~Theil, \emph{Validity and failure of the {Cauchy}-{Born}
  hypothesis in a two-dimensional mass-spring lattice}, J.\ Nonlinear Sci.
  \textbf{12} (2002), 445--478.

\bibitem{G2020}
J. Guilleminot, \emph{Modeling non-gaussian random fields of material
  properties in multiscale mechanics of materials}, Uncertainty Quantification
  in Multiscale Materials Modeling (Yan Wang and David~L. McDowell, eds.),
  Elsevier Series in Mechanics of Advanced Materials, Woodhead Publishing,
  2020, pp.~385--420.

\bibitem{guill_soize1}
J. Guilleminot and C. Soize, \emph{{On the statistical dependence
  for the components of random elasticity tensors exhibiting material symmetry
  properties}}, {Journal of Elasticity} \textbf{111} (2013), no.~2, 109--130.

\bibitem{GS2020}
J. Guilleminot and C. Soize, \emph{Non-gaussian random fields in
  multiscale mechanics of heterogeneous materials}, Encyclopedia of Continuum
  Mechanics (Holm Altenbach and Andreas {\"O}chsner, eds.), Springer Berlin
  Heidelberg, Berlin, Heidelberg, 2020, pp.~1826--1834.

\bibitem{HTB90}
P.~H\"anggi, P.~Talkner, and M.~Borkovec, \emph{Reaction-rate theory: fifty
  years after {Kramers}}, Rev. Mod. Phys. \textbf{62} (1990), 251--341.

%\bibitem{H17}
%T.~Hudson, \emph{Upscaling a model for the thermally-driven motion of screw
%  dislocations}, Arch. Ration. Mech. Anal. \textbf{224} (2017), no.~1,
%  291--352.

\bibitem{Jaynes1957}
E.~T. Jaynes, \emph{Information theory and statistical mechanics}, Phys. Rev.
  \textbf{106} (1957), 620--630.

\bibitem{Kermode2008}
J~R Kermode, T~Albaret, Dov Sherman, Noam Bernstein, P~Gumbsch, M~C Payne,
  G{\'a}bor Cs{\'a}nyi, and A~De~Vita, \emph{Low-speed fracture instabilities
  in a brittle crystal}, Nature \textbf{455} (2008), no.~7217, 1224--1227.

\bibitem{Kittel1996}
C.~Kittel, P.~McEuen, and P.~McEuen, \emph{Introduction to solid state
  physics}, vol.~8, Wiley New York, 1996.

\bibitem{landau1989theory}
L.D. Landau, E.M. Lifshitz, and J.B. Sykes, \emph{Theory of elasticity}, Course
  of theoretical physics, Pergamon Press, 1989.

\bibitem{LJ1924}
J.~E. Lennard-Jones, \emph{On the determination of molecular fields.
  {III}.{\textemdash}from crystal measurements and kinetic theory data},
  Proceedings of the Royal Society of London. Series A, Containing Papers of a
  Mathematical and Physical Character \textbf{106} (1924), no.~740, 709--718.

\bibitem{Longbottom2019}
S.~Longbottom and P.~Brommer, \emph{Uncertainty quantification for
  classical effective potentials: an extension to potfit}, Modelling and
  Simulation in Materials Science and Engineering \textbf{27} (2019), no.~4,
  044001.

\bibitem{Mehrabadi:1990:MC} 
M.~M.~Mehrabadi and S.~C.~ Cowin SC, \emph{Eigentensors of linear anisotropic elastic materials}, The Quarterly Journal of Mechanics and Applied Mathematics \textbf{43}(1990), 15--41 (doi: 10.1093/qjmam/43.1.15).

\bibitem{Mehta2004}
M.~L. Mehta, \emph{Random matrices}, Elsevier, 2004.

\bibitem{ortnertheil13}
C.~Ortner and F.~Theil, \emph{Justification of the {Cauchy}-{Born}
  approximation of elastodynamics}, Arch.\ Ration. Mech.\ Anal. \textbf{207}
  (2013), 1025--1073.

\bibitem{Ostoja2002}
M.~Ostoja-Starzewski, \emph{Lattice models in micromechanics}, Appl. Mech.
  Rev. \textbf{55} (2002), no.~1, 35--60.

\bibitem{Sinclair}
 J.~E. Sinclair \emph{The influence of the interatomic force law and of kinks on the propagation of brittle cracks}, The Philosophical Magazine: A Journal of Theoretical Experimental and Applied Physics, \textbf{31} (1975), 647-671.

\bibitem{Soize2006}
C.~Soize, \emph{Non-gaussian positive-definite matrix-valued random fields for
  elliptic stochastic partial differential operators}, Computer Methods in
  Applied Mechanics and Engineering \textbf{195} (2006), no.~1, 26--64.

\bibitem{Soize2017}
Christian Soize, \emph{Uncertainty quantification: An accelerated course with
  advanced applications in computational engineering}, vol.~47, Springer, 2017.

\bibitem{SJ12}
C.T. Sun and Z.-H. Jin, \emph{Fracture mechanics}, Academic Press, 2012.

\bibitem{Thomson_1971}
R.~Thomson, C.~Hsieh, and V.~Rana, \emph{Lattice trapping of fracture cracks},
  Journal of Applied Physics \textbf{42} (1971), no.~8, 3154--3160.

\bibitem{Thorpe1992}
M.~F. Thorpe and I.~Jasiuk, \emph{New results in the theory of elasticity for
  two-dimensional composites}, Proceedings: Mathematical and Physical Sciences
  \textbf{438} (1992), no.~1904, 531--544.
  
  \bibitem{BifKit}
Romain Veltz, \emph{{BifurcationKit.jl}},
  \url{https://hal.archives-ouvertes.fr/hal-02902346}, July 2020.

\bibitem{Wen2020}
Mingjian Wen and Ellad~B Tadmor, \emph{Uncertainty quantification in molecular
  simulations with dropout neural network potentials}, npj Computational
  Materials \textbf{6} (2020), no.~1, 1--10.

\bibitem{Zehnder2012}
Alan~T. Zehnder, \emph{Fracture mechanics}, Springer Netherlands, 2012.

\end{thebibliography}

%plain bibliography in a seprate tex file
%%% some useful commands for the bibliography:
  \newcommand{\noop}[1]{}
\providecommand{\bysame}{\leavevmode\hbox to3em{\hrulefill}\thinspace}
\providecommand{\MR}{\relax\ifhmode\unskip\space\fi MR }
\providecommand{\MRhref}[2]{%
  \href{http://www.ams.org/mathscinet-getitem?mr=#1}{#2}
}
\providecommand{\href}[2]{#2}

\end{document}